\documentclass[final,12pt]{colt2025} 

\title{Improved Algorithms for Effective Resistance Computation on Graphs}
\usepackage{times}
\usepackage{booktabs}
\usepackage{multirow}
\usepackage{algorithm,algorithmic}
\usepackage{tikz}
\usepackage{hyperref}
\hypersetup{
    colorlinks=true, 
    linkcolor=blue,  
    citecolor=blue,  
    urlcolor=blue    
}
\usepackage{amsfonts}
\usetikzlibrary{shapes.geometric}

\newcommand{\downlink}[1]{%
  \hyperlink{#1}{$\downarrow$}%
}

\coltauthor{%
 \Name{Yichun Yang} \Email{\href{mailto:yc.yang@bit.edu.cn}{yc.yang@bit.edu.cn}}\\
 \addr Beijing Institute of Technology
 \AND
 \Name{Rong-Hua Li} \Email{\href{mailto:lironghuabit@126.com}{lironghuabit@126.com}}\\
 \addr Beijing Institute of Technology%
 \AND
 \Name{Meihao Liao} \Email{\href{mailto:mhliao@bit.edu.cn}{mhliao@bit.edu.cn}}\\
 \addr Beijing Institute of Technology%
 \AND
 \Name{Guoren Wang} \Email{\href{mailto:wanggrbit@gmail.com}{wanggrbit@gmail.com}}\\
 \addr Beijing Institute of Technology%
}

\begin{document}

\maketitle
\newcommand{\ignore}[1]{}
\newcommand{\nop}[1]{}
\newcommand{\eat}[1]{}
\newcommand{\kw}[1]{{\ensuremath{\mathsf{#1}}}\xspace}
\newcommand{\kwnospace}[1]{{\ensuremath {\mathsf{#1}}}}
\newcommand{\stitle}[1]{\vspace{1ex} \noindent{\bf #1}}
\long\def\comment#1{}
\newcommand{\eop}{\hspace*{\fill}\mbox{$\Box$}}

\newtheorem{fact}{Fact}
\newtheorem{assumption}{Assumption}

\newcommand{\rank}{\kw{rank}}
\newcommand{\push}{\kw{CGD}}
\newcommand{\truncatepush}{\kw{Push}}
\newcommand{\hkrelax}{\kw{Hk\ Relax}}
\newcommand{\hkpush}{\kw{Hk\ Push}}
\newcommand{\agp}{\kw{AGP}}
\newcommand{\tea}{\kw{TEA}}
\newcommand{\teaplus}{\kw{TEA+}}
\newcommand{\ppr}{\kw{PPR}}
\newcommand{\ssppr}{\kw{SSPPR}}
\newcommand{\hkpr}{\kw{HKPR}}
\newcommand{\powerpush}{\kw{PwPush}}
\newcommand{\pwpush}{\kw{PowerPush}}
\newcommand{\powerpushsor}{\kw{PwPushSOR}}
\newcommand{\ltwocheb}{\kw{ChebyPower}}

\newcommand{\chebpush}{\kw{ChebyPush}}
\newcommand{\powermethod}{\kw{PM}}
\newcommand{\lanczos}{\kw{Lz}}
\newcommand{\lzpush}{\kw{LzPush}}
\newcommand{\bipush}{\kw{BiPush}}
\newcommand{\geer}{\kw{GEER}}
\newcommand{\rw}{\kw{RW}}
\newcommand{\lv}{\kw{LV}}
\newcommand{\lewalk}{\kw{LE\textrm{-}Walk}}
\newcommand{\bippr}{\kw{Bippr}}
\newcommand{\fora}{\kw{FORA}}
\newcommand{\speedppr}{\kw{SpeedPPR}}
\newcommand{\foralv}{\kw{FORALV}}
\newcommand{\speedlv}{\kw{SpeedLV}}
\newcommand{\setpush}{\kw{SetPush}}
\newcommand{\chopper}{\kw{CHOPPER}}

\newcommand{\dblp}{\kw{Dblp}}
\newcommand{\asskitter}{\kw{As\textrm{-}Skitter}}
\newcommand{\orkut}{\kw{Orkut}}
\newcommand{\youtube}{\kw{Youtube}}
\newcommand{\livejournal}{\kw{LiveJournal}}
\newcommand{\roadca}{\kw{RoadNet\textrm{-}CA}}
\newcommand{\roadpa}{\kw{RoadNet\textrm{-}PA}}
\newcommand{\roadtx}{\kw{RoadNet\textrm{-}TX}}
\newcommand{\powergrid}{\kw{powergrid}}
\newcommand{\pokec}{\kw{Pokec}}
\newcommand{\twitter}{\kw{Twitter}}
\newcommand{\friendster}{\kw{Friendster}}

\begin{abstract}

Effective Resistance (ER) is a fundamental tool in various graph learning tasks. In this paper, we address the problem of efficiently approximating ER on a graph $\mathcal{G}=(\mathcal{V},\mathcal{E})$ with $n$ vertices and $m$ edges. First, we focus on local online-computation algorithms for ER approximation, aiming to improve the dependency on the approximation error parameter $\epsilon$. Specifically, for a given vertex pair $(s,t)$, we propose a local algorithm with a time complexity of $\tilde{O}(\sqrt{d}/\epsilon)$ to compute an $\epsilon$-approximation of the $s,t$-ER value for expander graphs, where $d=\min \{d_s,d_t\}$. This improves upon the previous state-of-the-art, including an $\tilde{O}(1/\epsilon^2)$ time algorithm based on random walk sampling by Andoni et al. (ITCS'19) and Peng et al. (KDD'21). Our method achieves this improvement by combining deterministic search with random walk sampling to reduce variance. Second, we establish a lower bound for ER approximation on expander graphs. We prove that for any $\epsilon\in (0,1)$, there exist an expander graph and a vertex pair $(s,t)$ such that any local algorithm requires at least $\Omega(1/\epsilon)$ time to compute the $\epsilon$-approximation of the $s,t$-ER value. Finally, we extend our techniques to index-based algorithms for ER computation. We propose an algorithm with $\tilde{O}(\min \{m+n/\epsilon^{1.5},\sqrt{nm}/\epsilon\})$ processing time, $\tilde{O}(n/\epsilon)$ space complexity and $O(1)$ query complexity, which returns an $\epsilon$-approximation of the $s,t$-ER value for any $s,t\in \mathcal{V}$ for expander graphs. Our approach improves upon the state-of-the-art $\tilde{O}(m/\epsilon)$ processing time by Dwaraknath et al. (NeurIPS'24) and the $\tilde{O}(m+n/\epsilon^2)$ processing time by Li and Sachdeva (SODA'23).

\end{abstract}


\section{Introduction}\label{sec:intro}
Given an undirected, unweighted graph $\mathcal{G}$\footnote{Though this paper primarily focus on unweighted graphs, all the algorithms and analysis can be directly generalized to weighted graphs with $\min_{e}{w_e}=\tilde{\Omega}(1)$ and $\max_e{w_e}=\tilde{O}(1)$, where $w_e$ denotes the weight of an edge $e$.}, the Effective Resistance (ER) $r_{\mathcal{G}}(s,t)$ of two vertices $s$ and $t$ is the potential difference between $s$ and $t$, when a unit electric flow is sent from $s$ to $t$. Formally, given the graph Laplacian matrix $\mathbf{L}$ and its pseudo-inverse $\mathbf{L}^\dagger$, the $s,t$ ER value is defined as $r_{\mathcal{G}}(s,t)=(\mathbf{e}_s-\mathbf{e}_t)^T\mathbf{L}^\dagger(\mathbf{e}_s-\mathbf{e}_t)$, where $\mathbf{e}_s$ (resp., $\mathbf{e}_t$) denotes the one-hot vector that takes value $1$ at $s$ (resp., $t$) and $0$ elsewhere. ER has a wide range of applications in graph algorithms and graph learning tasks, including the design of maximum flow algorithms ~\cite{van2022faster,madry2016computing}, graph clustering ~\cite{alev2017graph,saito2023multi}, graph sparsification ~\cite{spielman2008graph}, counting random spanning trees ~\cite{li2023new} and understanding oversquashing in GNNs ~\cite{black2023understanding,topping2021understanding}. As a result, designing efficient algorithms for ER computation has become a fundamental problem, attracting significant attention in recent years ~\cite{andoni2018solving,peng2021local,yang2023efficient,liao2023resistance,li2023new,dwaraknath2024towards}. In this paper, we focus on computing an $\epsilon$-approximation of ER on a given set of vertex pairs $S\subseteq \mathcal{V}\times \mathcal{V}$. We begin by formally defining the ER approximation problem.

\begin{definition}
    Given undirected, unweighted graph $\mathcal{G}=(\mathcal{V},\mathcal{E})$, a set of vertex pairs $S\subseteq \mathcal{V}\times \mathcal{V}$, a parameter $\epsilon\in (0,1)$, one requires to find an approximation $\hat{r}_{\mathcal{G}}\in \mathbb{R}^S$ such that $ \hat{r}_{\mathcal{G}}(s,t)\approx_\epsilon r_{\mathcal{G}}(s,t)$\footnote{In this paper, we say $x\approx_\epsilon y$ if $(1-\epsilon)y\leq x\leq (1+\epsilon)y$.} with high probability (w.h.p.) for any vertex pairs $(s,t)\in S$. We refer to such $\hat{r}_{\mathcal{G}}\in \mathbb{R}^S$ as an $\epsilon$-approximation of $r_{\mathcal{G}}$ on $S$. 
\end{definition}

The existing algorithms for computing $\epsilon$-approximation of ER can be broadly classified into two categories: (i) online-computation algorithms; and (ii) index-based algorithms. Specifically, for the given graph $\mathcal{G}$, online-computation algorithms directly calculates the $s,t$-ER value of a given vertex pair $s,t$. In contrast, index-based algorithms preprocess the graph to construct a data structure, which is then used to efficiently query the $s,t$-ER value for any given vertex pair  $s,t$. The runtime of the state-of-the-art algorithms are summarized in Table \ref{tab:alg_single_pair} and Table \ref{tab:alg_multiple_pair}. In this paper, we assume all the algorithms support the classic adjacency graph model ~\cite{ron2019sublinear}, which allows the following three types of queries on a graph $\mathcal{G}=(\mathcal{V},\mathcal{E})$ in constant time: 

\begin{itemize}

\item Degree query (given any $u\in \mathcal{V}$, get the degree $deg(u)$ in constant time); 

\item Neighbor query (given any $u\in \mathcal{V}$, get its $i$-th neighbor $Neigh(u,i)$ in constant time); 

\item Jump query (uniformly sample any $u\in \mathcal{V}$ in constant time).
\end{itemize}

\stitle{Online-computation algorithms.} We begin by focusing on online-computation algorithms to calculate ER value of a single vertex pair. For a given graph $\mathcal{G}$, given a vertex pair $(s,t)$, by the definition $r_{\mathcal{G}}(s,t)=(\mathbf{e}_s-\mathbf{e}_t)^T\mathbf{L}^\dagger(\mathbf{e}_s-\mathbf{e}_t)$. Using a nearly linear-time Laplacian solver ~\cite{sachdeva2014faster}, one can compute an approximation $\hat{r}_{\mathcal{G}}(s,t)$ such that $ \hat{r}_{\mathcal{G}}(s,t)\approx_\epsilon r_{\mathcal{G}}(s,t)$ in $O(m \log ^{1/2} n \log \frac{1}{\epsilon})$ time, where $n$ is the number of vertices, $m$ is the number of edges, $\epsilon$ is the approximation error parameter. In recent years, there has been growing interests in the local online-computation of ER, which aims to compute $ \hat{r}_{\mathcal{G}}(s,t)\approx_\epsilon r_{\mathcal{G}}(s,t)$ by just exploring a small portion of the graph. 
These algorithms typically rely on the assumption that the graph is an expander.  ~\cite{andoni2018solving} proposed an algorithm that computes an $\epsilon$-approximation of  $r_{\mathcal{G}}(s,t)$ in $\tilde{O}(1/\epsilon^2)$\footnote{The $\tilde{O}(.)$ notation represents the complexity omitting the $\log$ term of $n,m$ and $\epsilon$.} time for $d$-regular expander graphs. This result was later extended to expander graphs with unbounded degrees by \cite{peng2021local} and ~\cite{yang2023efficient}. However, it remains an open question whether the $\tilde{O}(1/\epsilon^2)$ runtime can be improved for $\epsilon$-approximating the ER value for local algorithms. In this paper, we address this problem by proposing a local algorithm for ER computation that achieves a lower time complexity dependency on $\epsilon$. The key idea of our approach is to combine deterministic search with random walks to reduce the variance produced by random walk sampling.

\begin{theorem}[\downlink{proof_of_theorem:bidir_complexity}]\label{thm:bidir_complexity}
    Given an undirected, unweighted expander graph $\mathcal{G}=(\mathcal{V},\mathcal{E})$, $s$, $t\in \mathcal{V}$, $\epsilon\in (0,1)$. There exists an algorithm that outputs the approximation $ \hat{r}_{\mathcal{G}}(s,t)\approx_\epsilon r_{\mathcal{G}}(s,t)$ in $\tilde{O}(\sqrt{d}/ \epsilon)$ time, where $d=\min \{d_s,d_t\}$.
\end{theorem}

There is an additional remark that without the assumption on expander graphs, the time complexity of our algorithm is $\tilde{O}(\kappa(\mathcal{L})^3 \sqrt{d}/\epsilon)$, where $\kappa(\mathcal{L})$ is the condition number of the normalized Laplacian matrix $\mathcal{L}$ (see Section \ref{sec:preliminary} for Definition). This condition number dependency is the same as the previous local algorithms ~\cite{andoni2018solving,peng2021local,yang2023efficient}.

\begin{table*}[t!]
	\centering
	\caption{ER computation of single vertex pair} \label{tab:alg_single_pair} 
	\scalebox{1}{
		\begin{tabular}{c c c}
			\toprule
			\multicolumn{1}{c}{Methods} & \multicolumn{1}{c}{Runtime} & Assumption\\
			\midrule
             ~\cite{sachdeva2014faster}& $O(m \log ^{1/2} n \log \frac{1}{\epsilon})$ & none\\
            ~\cite{andoni2018solving}& $\tilde{O}(1/\epsilon^2)$ & expanders\\
            ~\cite{peng2021local}& $\tilde{O}(1/\epsilon^2)$ & expanders\\
            ~\cite{yang2023efficient}& $\tilde{O}(1/\epsilon^2)$ & expanders\\
            This Paper& $\tilde{O}(\sqrt{d}/\epsilon)$ & expanders\\
            \midrule
             Lower Bound (This Paper) & $\Omega(1/\epsilon)$ & expanders\\
            \bottomrule	
		\end{tabular}
	}
\end{table*}

\stitle{Lower bounds.} We also study the lower bound for the ER computation of a single pair. Specifically, we prove that the  time complexity of any local online-computation algorithm for approximating ER is at least $\Omega(1/\epsilon)$, even holds for  expander graphs. This result implies that the $1/\epsilon$ dependency on the approximation parameter $\epsilon$ is the best possible, and the gap between our upper bound and the lower bound is only a factor $\sqrt{d}$.

\begin{theorem}[\downlink{proof_of_theorem:lower_bound}]\label{thm:lower_bound}
    Given any $ \epsilon\in (0,1)$, there exists an expander graph $\mathcal{G}=(\mathcal{V},\mathcal{E})$, two vertices $s$ and $t$, such that for any (randomized) local algorithm that supports the adjacency model computes the approximation $ \hat{r}_{\mathcal{G}}(s,t)\approx_\epsilon r_{\mathcal{G}}(s,t)$ with success probability $\geq 2/3$ requires $\Omega(1/\epsilon)$  queries.
\end{theorem}

We note that the conditional lower bounds by ~\cite{dwaraknath2024towards} implies a $\Omega(n^{2}/\sqrt{\epsilon})$ lower bound for all pairs ER approximation (i.e., $\epsilon$-approximates $r_{\mathcal{G}}(s,t)$ for all $s,t\in \mathcal{V}$) for non-fast matrix multiplication (FMM) algorithms. This further implies a $\Omega(1/\sqrt{\epsilon})$ lower bound for ER approximation of a single pair. However, by Theorem \ref{thm:lower_bound}, we prove a stronger lower bound $\Omega(1/\epsilon)$. Furthermore, our construction is simpler and does not rely on the non-FMM assumption.

\stitle{Index-based algorithms.} We show that our techniques can be generalized to index-based ER computations. Following ~\cite{li2023new,dwaraknath2024towards}, we define the ER sketch algorithm as follows.

\begin{definition}{(ER sketch)}
    Given undirected, unweighted $\mathcal{G}=(\mathcal{V},\mathcal{E})$ and $\epsilon\in (0,1)$, We call a (randomized) algorithm an $(T_p,T_q,S_t)$-ER sketch algorithm if in $O(T_p)$ processing time it creates a binary string of length $O(S_t)$, from which when queried with $\forall s,t\in \mathcal{V}$, it outputs the approximation $\hat{r}_{\mathcal{G}}(s,t)$ such that $ \hat{r}_{\mathcal{G}}(s,t)\approx_\epsilon r_{\mathcal{G}}(s,t)$ w.h.p. in $O(T_q)$ time.
\end{definition}

Note that this ER sketch problem is closely related to the computation of graph Laplacian sparsifiers, which has been extensively studied in previous studies~\cite{spielman2008graph,DKP+2017spanningtree,jambulapati2018efficient,chu2020graph}. The current main open question is how to improve the processing time dependency on $n,m$ and $\epsilon$. The state-of-the-art methods for this ER sketch problem on expander graphs include an algorithm with $(\tilde{O}(m+n/\epsilon^{2}),O(1),\tilde{O}(n/\epsilon))$ complexity based on random walk sampling ~\cite{li2023new}  and an algorithm with $(\tilde{O}(m/\epsilon), \tilde{O}(1),\tilde{O}(n/\epsilon))$ complexity based on Count Sketch and Laplacian solver ~\cite{dwaraknath2024towards}. In this paper, we show that we can also improve the runtime for the ER sketch problem by combining deterministic search with random walk sampling.

\begin{theorem}[\downlink{proof_of_theorem:ER_sketch}]\label{thm:ER_sketch}
    There is an $(\tilde{O}(\min \{m+n/\epsilon^{1.5}, \sqrt{nm}/\epsilon\}),O(1),\tilde{O}(n/\epsilon))$ ER sketch algorithm for undirected, unweighted expander graph $\mathcal{G}=(\mathcal{V},\mathcal{E})$, and $\epsilon\in (0,1)$. 
\end{theorem}

From Theorem \ref{thm:ER_sketch}, we observe that the processing time of our ER sketch algorithm is $T_p=\tilde{O}(\min \{m+n/\epsilon^{1.5}, \sqrt{nm}/\epsilon\})$. For the first term $m+n/\epsilon^{1.5}$, this represents a speedup of $1/\sqrt{\epsilon}$ compared to the $\tilde{O}(m+n/\epsilon^{2})$ algorithm proposed in ~\cite{li2023new}. For the second term $\sqrt{nm}/\epsilon$, this achieves a $\sqrt{m/n}$ speedup over the $\tilde{O}(m/\epsilon)$ algorithm proposed in ~\cite{dwaraknath2024towards}. Consequently, our algorithm establishes an improved runtime bound compared to the state-of-the-art methods. There is an additional remark that without the assumption on expander graphs, the processing time of our algorithm is $\tilde{O}(\min \{m+\kappa(\mathcal{L})^3n/\epsilon^{1.5}, \kappa(\mathcal{L})^3\sqrt{nm}/\epsilon\})$, the space complexity is $\tilde{O}(\kappa(\mathcal{L})n/\epsilon)$, and the query time remains $O(1)$. This condition number dependency is the same as the random walk-based method ~\cite{li2023new}.

\begin{table*}[t!]
	\centering
	\caption{ER computation of multiple vertex pairs $S\subset \mathcal{V}\times \mathcal{V}$} \label{tab:alg_multiple_pair} 
	\scalebox{1}{
		\begin{tabular}{c c c}
			\toprule
            \multicolumn{1}{c}{Methods}&\multicolumn{1}{c}{Runtime}& Assumption\\
			\midrule
              ~\cite{spielman2008graph}&	$\tilde{O}(m/\epsilon^2+|S|/\epsilon^2)$&none\\
              ~\cite{DKP+2017spanningtree}&	$\tilde{O}(m+n/\epsilon^2+|S|/\epsilon^2)$&none\\
              ~\cite{jambulapati2018efficient}&	$\tilde{O}(n^2/\epsilon)$&none\\
			 ~\cite{chu2020graph}& $\tilde{O}(m+n/\epsilon^{1.5}+|S|/\epsilon^{1.5})$&none\\
              ~\cite{li2023new}& $\tilde{O}(m+n/\epsilon^{2}+|S|)$&expanders\\
              ~\cite{dwaraknath2024towards}&  $\tilde{O}(m/\epsilon+|S|)$&expanders\\
            
             This Paper& $\tilde{O}(\min \{m+n/\epsilon^{1.5}, \sqrt{nm}/\epsilon\}+|S|)$& expanders\\
            \midrule
             Lower Bound ~\cite{dwaraknath2024towards}&  $\Omega(n^{c_1}/\epsilon^{c_2})$ with $c_1+2c_2=3$ & $S=\mathcal{V}\times \mathcal{V}$\\
            \bottomrule	
		\end{tabular}
	}
\end{table*}

\subsection{Technique overview}

In this paper, we primarily focus on the Taylor expansion of $\mathbf{L}^\dagger (\mathbf{e}_s-\mathbf{e}_t)$ to compute the ER value $r_{\mathcal{G}}(s,t)$. Based on this Taylor expansion representation, existing local algorithms sampling random walks to approximate $\mathbf{L}^\dagger (\mathbf{e}_s-\mathbf{e}_t)$ to $\epsilon$-approximates $r_{\mathcal{G}}(s,t)$ ~\cite{andoni2018solving,peng2021local}. Our main algorithmic contribution is to combine deterministic search with random walk sampling, thereby reducing the variance introduced by random walks. This approach is inspired by the two-phase algorithm for local PageRank computation ~\cite{lofgren16bidirection,wang2024revisiting,wei2024approximating}. Specifically, we use a variant of the classical coordinate gradient descent algorithm as the deterministic part to coarsely approximate $\mathbf{L}^\dagger (\mathbf{e}_s-\mathbf{e}_t)$. Subsequently, we sample random walks to refine the residuals produced by the coordinate gradient descent. We carefully analyze the error bounds and complexity of our proposed algorithm, proving that the runtime dependency on $\epsilon$ can be improved through the combination of coordinate gradient descent and random walk sampling. For the index-based algorithm, we have the following two key observations: (i) the approximation of $\mathbf{L}^\dagger (\mathbf{e}_s-\mathbf{e}_t)$ is a sparse vector; (ii) random walks can share the computation of the deterministic part (i.e., coordinate gradient descent). Based on these two observations, we extend our techniques and design an index-based algorithm faster than the state-of-the-art.

Additionally, we study the lower bound for the local computation of single pair ER on expander graphs. Our key observation is the sensitivity of the parallel resistors. Roughly speaking, we consider the following two cases: 
(i) case $C_1$: we connect $d$ resistors with resistance $1$ in parallel; 
(ii) case $C_2$: we connect $(1-\epsilon)d$ resistors with resistance $1$ and $\epsilon d$ resistors with resistance $2$ in parallel. 
We prove that the effective resistance of $C_1$ and $C_2$ differs by $\Omega(\epsilon)$ in terms of relative error. However, distinguishing $C_1$ and $C_2$ is challenging for local algorithms. This observation leads to the construction of our lower bound.

\subsection{Related Work}



\stitle{Local graph algorithms.} Local graph algorithms for several related problems have been extensively studied. For instance, ~\cite{liu2020strongly,fountoulakis2019variational,fountoulakis2020p,fountoulakis2023flow} investigate local algorithms for conductance-based graph clustering, while ~\cite{lofgren2013personalized,lofgren16bidirection,bressan2018sublinear,wang2024revisiting,wei2024approximating} explore local algorithms for PageRank computation. Additionally, ~\cite{cohen2018approximating,jin2023moments} examine local algorithms for graph spectral density approximation, and ~\cite{andoni2018solving} study local algorithms to solve Laplacian systems. The problem setting of our work, which focuses on the local computation of ER for a single-pair vertices, shares similarities with these problems. However, the techniques developed for these related problems cannot be applied directly in our setting.

\stitle{Lower bounds for ER computation.} For the local computation of ER, in addition to our lower bound concerning the approximation parameter $\epsilon$, previous studies have also focused on establishing lower bounds related to the condition number $\kappa(\mathcal{L})$. \cite{andoni2018solving} proved that any local algorithm that $\epsilon$-approximates a coordinate of the Laplacian linear system requires at least $\Omega(\kappa(\mathcal{L})^2)$ queries, when setting $\epsilon=\Theta(1/\log n)$. However, since their problem setting differs from our local ER computation problem, their results cannot be directly generalized to our problem. In addition, ~\cite{cai2023effective} provides a lower bound for the single-pair ER approximation problem for non-expander graphs. They proved that there exists a non-expander graph $\mathcal{G}$ with minimum degree $\delta\geq 3$, and an adjacent vertex pair $(s,t)\in \mathcal{E}$, such that any local algorithm requires $\Omega (n)$ queries to $\epsilon$-approximates $r_\mathcal{G}(s,t)$ for any $\epsilon\leq 0.01$. It can be easily proved that the condition number of their construction is $\kappa(\mathcal{L})=\Theta(n)$, thus implies the $\Omega (\kappa(\mathcal{L}))$ lower bound for the $\epsilon$-approximation of $s,t$-ER value, for a given $\epsilon\leq 0.01$. However, it remains an open problem to study the relationship between $\epsilon$ and $\kappa(\mathcal{L})$ for local algorithms, so as to unify the different cases for expanders and non-expanders (since for expanders $\kappa(\mathcal{L})=\tilde{O}(1)$).

\stitle{Applications of ER computation.} Since ER approximation algorithms can be used as a crucial subroutine for other graph analysis algorithms, many related studies have focused on ER approximation and its applications. Among them, recent advancements in graph clustering and counting random spanning trees are closely related to our results. For instance, ~\cite{alev2017graph} proposed the ER-based graph clustering algorithm in $\tilde{O}(nm)$ time with non-trivial theoretical guarantees. 
~\cite{chu2020graph} design an $\tilde{O}(m+n^{1.875}/\delta^{1.75})$ time algorithm that computes a $(1+\delta)$-approximation of the number of spanning trees. This result was later improved by ~\cite{li2023new} to $\tilde{O}(m+n^{1.5}/\delta)$, achieving optimality for determinant sparsifier-based methods in the case of expander graphs. These works leverage the ER sketch algorithm as a key subroutine. Beyond these applications, ER computation has also been utilized in graph sparsification ~\cite{spielman2008graph} and robust routing in road networks ~\cite{sinop2023robust}, among others. The potential impact of our algorithms on enhancing the efficiency of these applications remains an open direction for future research.

\section{Preliminaries}\label{sec:preliminary}


\stitle{General matrix notations.} Given a matrix $\mathbf{X}\in \mathbb{R}^{n\times n}$, we use $\lambda(\mathbf{X})$ to denote its spectrum, $\lambda_i(\mathbf{X})$ represents the $i$-th eigenvalue of $\mathbf{X}$, $\mathbf{u}_i$ denotes the corresponding eigenvector. For a positive semi-definite (PSD) matrix $\mathbf{X}$, the condition number is defined as $\kappa(\mathbf{X})=\frac{\lambda_{max}(\mathbf{X})}{\lambda_{min}(\mathbf{X})}$, where $\lambda_{max}(\mathbf{X})$ is the maximum eigenvalue of $\mathbf{X}$, $\lambda_{min}(\mathbf{X})$ is the minimum non-zero eigenvalue of $\mathbf{X}$. Note that for singular matrix, this represents the finite condition number instead of the original definition $\kappa(\mathbf{X})=\infty$. Let $\mathbf{e}_s\in \mathbb{R}^n$ be the one-hot vector, which takes value $1$ at $s\in [n]$ and $0$ elsewhere.

\stitle{Graph notations.} Given an undirected, unweighted graph $\mathcal{G}=(\mathcal{V},\mathcal{E})$, we denote $\mathcal{V}$ as the vertex set and $\mathcal{E}$ as the edge set. Let $|\mathcal{V}|=n$ be the number of vertices and $|\mathcal{E}|=m$ be the number of edges of $\mathcal{G}$. For a node $u\in \mathcal{V}$, the set of neighbors is denoted as $\mathcal{N}(u)=\{v|(u,v)\in \mathcal{E}\}$, and the degree of node $u$ is denoted as $d_u=|\mathcal{N}(u)|$. Let $\mathbf{D}$ be the diagonal degree matrix, with diagonal entry $\mathbf{D}_{i,i}=d_i$ (the degree of node $i$). Denote by $\mathbf{A}$ the adjacency matrix with $\mathbf{A}_{i,j}=1$ if and only if $(i,j)\in \mathcal{E}$. Then, the Laplacian matrix is defined as $\mathbf{L}=\mathbf{D}-\mathbf{A}$. Let $\mathcal{A}=\mathbf{D}^{-1/2}\mathbf{A}\mathbf{D}^{-1/2}$ be the normalized adjacency matrix, $\mathbf{P}=\mathbf{AD}^{-1}$ be the probability transition matrix, and $\mathcal{L}=\mathbf{D}^{-1/2}\mathbf{L}\mathbf{D}^{-1/2}=\mathbf{I}-\mathcal{A}$ be the normalized Laplacian matrix. We denote $0=\lambda_1(\mathcal{L})< \lambda_2(\mathcal{L})\leq...\leq \lambda_n(\mathcal{L})\leq 2$ as the eigenvalues of $\mathcal{L}$, and $1=\lambda_1(\mathcal{A})> \lambda_2(\mathcal{A})\geq...\geq \lambda_n(\mathcal{A})\geq -1$ as the eigenvalues of $\mathcal{A}$. By definition, we have $\lambda_i(\mathcal{A})=1-\lambda_i(\mathcal{L})$ for each $i\in[n]$.
Following ~\cite{dwaraknath2024towards,cai2023effective}, we define the expander graphs based on conductance.

\begin{definition}{(Expander)}
    We define the conductance of a graph $\mathcal{G}$ as:
    $$\phi_\mathcal{G}=\min_{S\subset \mathcal{V}} {\frac{|\{(u,v)\in \mathcal{E}: u\in S,v\in \mathcal{V}-S \}|}{\min \{\sum_{u\in S}{d_u},\sum_{u\in \mathcal{V}-S}{d_u}\}}}.$$
    
    We refer to the graph $\mathcal{G}$ as a expander if $\phi_\mathcal{G}=\tilde{\Omega}(1)$.
\end{definition}

By the classical Cheeger's inequality (that is, $\lambda_2(\mathcal{L})/2\leq \phi_\mathcal{G}\leq \sqrt{2\lambda_2(\mathcal{L})}$, see e.g. ~\cite{spielman2019sagt} Chapter 21), we have $\phi_\mathcal{G}=\tilde{\Omega}(1)$ if and only if $\lambda_2(\mathcal{L})=\tilde{\Omega}(1)$, where $\lambda_2(\mathcal{L})$ is the second smallest eigenvalue of the normalized Laplacian $\mathcal{L}$ (also the smallest non-zero eigenvalue). Consequently, the definition of an expander graph can be equivalently characterized by $\lambda_2(\mathcal{L})=\tilde{\Omega}(1)$  or by the condition number $\kappa(\mathcal{L})=\tilde{O}(1)$.

\section{Basic representation}\label{sec:basic_repre}

Following ~\cite{peng2021local,li2023new,dwaraknath2024towards}, we start by interpreting the ER value through the Taylor expansion of the Laplacian pseudo-inverse (but our interpretation is simpler and slightly differs from the previous works ~\cite{peng2021local,li2023new,dwaraknath2024towards}). We begin by stating the following lemma, which expresses $r_\mathcal{G}(s,t)$ as the sum of infinite step of lazy random walks. 

\begin{lemma}[\downlink{proof_of_lemma:ER_taylor_expansion}]\label{lemma:ER_taylor_expansion}
 The following Equation holds
    \begin{align*}
r_\mathcal{G}(s,t)&=(\mathbf{e}_s-\mathbf{e}_t)^T\mathbf{L}^{\dagger}(\mathbf{e}_s-\mathbf{e}_t)\\
&=\frac{1}{2}(\mathbf{e}_s-\mathbf{e}_t)^T\mathbf{D}^{-1}\sum_{l=0}^{+\infty}{\left(\frac{1}{2}\mathbf{I}+\frac{1}{2}\mathbf{P}\right)^l}(\mathbf{e}_s-\mathbf{e}_t)
\end{align*}
with $\mathbf{P}=\mathbf{AD}^{-1}$ denotes the probability transition matrix.
\end{lemma}

By the Taylor expansion of $r_\mathcal{G}(s,t)$ from Lemma \ref{lemma:ER_taylor_expansion}, we define the $L$-step truncated ER value $r_{\mathcal{G},L}(s,t)$ as follows.

\begin{definition}\label{def:truncate_ER}
We define
$$r_{\mathcal{G},L}(s,t)=\frac{1}{2}(\mathbf{e}_s-\mathbf{e}_t)^T\mathbf{D}^{-1}\sum_{l=0}^{L}{\left(\frac{1}{2}\mathbf{I}+\frac{1}{2}\mathbf{P}\right)^l}(\mathbf{e}_s-\mathbf{e}_t)$$
as the $L$ step truncation ER value.
\end{definition}

Next, we aim to ensure that $r_{\mathcal{G},L}(s,t)$ serves as an  $\epsilon$-approximation of $r_\mathcal{G}(s,t)$ by selecting a sufficiently large truncation step $L$. However, for expander graphs, we will prove that $L=\tilde{O}(1)$ is sufficient to achieve this approximation.

\begin{lemma}[\downlink{proof_of_lemma:ER_truncate}]\label{lem:ER_truncate}
   $r_{\mathcal{G},L}(s,t)$ is the $\epsilon$-approximation of $r_\mathcal{G}(s,t)$ when setting $L\geq 2\kappa(\mathcal{L}) \log \frac{n}{\epsilon}$.
\end{lemma}

By Lemma~\ref{lem:ER_truncate}, the truncation step is only required to be $L=\tilde{O}(1)$, since $\kappa(\mathcal{L})=\tilde{O}(1)$ for expander graphs. Inspired by Lemma \ref{lem:ER_truncate}, we define the following vector.

\begin{definition}\label{def:p_L}
$\mathbf{p}_{L,u}=\frac{1}{2}\sum_{l=0}^{L}{\left(\frac{1}{2}\mathbf{I}+\frac{1}{2}\mathbf{P}\right)^l}\mathbf{e}_u$ for any given $u\in \mathcal{V}$. 
\end{definition}

From Definition \ref{def:truncate_ER} and Definition \ref{def:p_L}, we have the following representation of $r_\mathcal{G}(s,t)$:
\begin{equation}\label{equ:ER_p_expression}
\begin{aligned}
    r_{\mathcal{G},L}(s,t)&=\frac{1}{2}(\mathbf{e}_s-\mathbf{e}_t)^T\mathbf{D}^{-1}\sum_{l=0}^{L}{\left(\frac{1}{2}\mathbf{I}+\frac{1}{2}\mathbf{P}\right)^l}(\mathbf{e}_s-\mathbf{e}_t)\\
    &=\mathbf{e}_s^T\mathbf{D}^{-1}\mathbf{p}_{L,s}-\mathbf{e}_t^T\mathbf{D}^{-1}\mathbf{p}_{L,s}-\mathbf{e}_s^T\mathbf{D}^{-1}\mathbf{p}_{L,t}+\mathbf{e}_t^T\mathbf{D}^{-1}\mathbf{p}_{L,t}\\
    &=\frac{\mathbf{p}_{L,s}(s)}{d_s}-\frac{\mathbf{p}_{L,s}(t)}{d_t}-\frac{\mathbf{p}_{L,t}(s)}{d_s}+\frac{\mathbf{p}_{L,t}(t)}{d_t}.
\end{aligned}
\end{equation}

Below, we give some basic properties of $\mathbf{p}_{L,u}$ and $r_\mathcal{G}(s,t)$ which will be frequently used in our analysis later.

\begin{fact}[\downlink{proof_of_fact:p_L_norm_1}]\label{fact:p_L_norm_1}
    $\Vert \mathbf{p}_{L,u} \Vert_1\leq \frac{L}{2}$ for any $u\in \mathcal{V}$.
\end{fact}

\begin{fact}[\downlink{proof_of_fact:p_L_transform}]\label{fact:p_L_transform}
    $\frac{\mathbf{p}_{L,u}(v)}{d_v}=\frac{\mathbf{p}_{L,v}(u)}{d_u}$ for any $u,v\in \mathcal{V}$.
\end{fact}

\begin{fact}[\downlink{proof_of_fact:ER_lower_bound}]\label{fact:ER_lower_bound}
    $r_\mathcal{G}(s,t)\geq \frac{1}{2}\left(\frac{1}{d_s}+\frac{1}{d_t}\right)$.
\end{fact}

\section{Online ER Computation Algorithms}

In this section, we present our online-computation algorithm for single-pair ER computation and prove Theorem \ref{thm:bidir_complexity}. By combining Eq.~(\ref{equ:ER_p_expression}) with Lemma \ref{lem:ER_truncate}, we only need to find the approximation $\hat{\mathbf{p}}_{L,u}(v)$ for each $u,v\in\{s,t\}$ to approximate the ER value $r_\mathcal{G}(s,t)$. The approximation is given by $\hat{r}_\mathcal{G}(s,t)=\frac{\hat{\mathbf{p}}_{L,s}(s)}{d_s}-\frac{\hat{\mathbf{p}}_{L,s}(t)}{d_t}-\frac{\hat{\mathbf{p}}_{L,t}(s)}{d_s}+\frac{\hat{\mathbf{p}}_{L,t}(t)}{d_t}$. Formally, we prove the following Theorem.

\begin{theorem}[\downlink{proof_of_theorem:bidir_guarantee_specific}]\label{thm:bidir_guarantee_specific}
    There exists an Algorithm (i.e. our Algorithm \ref{algo:bidir_single_pair}) outputs the approximation $\hat{\mathbf{p}}_{L,u}(v)$ of $\mathbf{p}_{L,u}(v)$ such that $|\mathbf{p}_{L,u}(v)-\hat{\mathbf{p}}_{L,u}(v)|/d_v\leq \frac{\epsilon}{d}$ for each $u,v\in\{s,t\}$ w.h.p., where $d=\min\{d_s,d_t\}$. Furthermore, the time complexity of this Algorithm can be bounded by $\tilde{O}(\frac{\sqrt{d}}{\epsilon})$.
\end{theorem}

\begin{figure}
    \centering
    \subfigure[Our online-computation algorithm for computing $\mathbf{p}_{L,s}(t)$]{
    \includegraphics[scale=0.35]{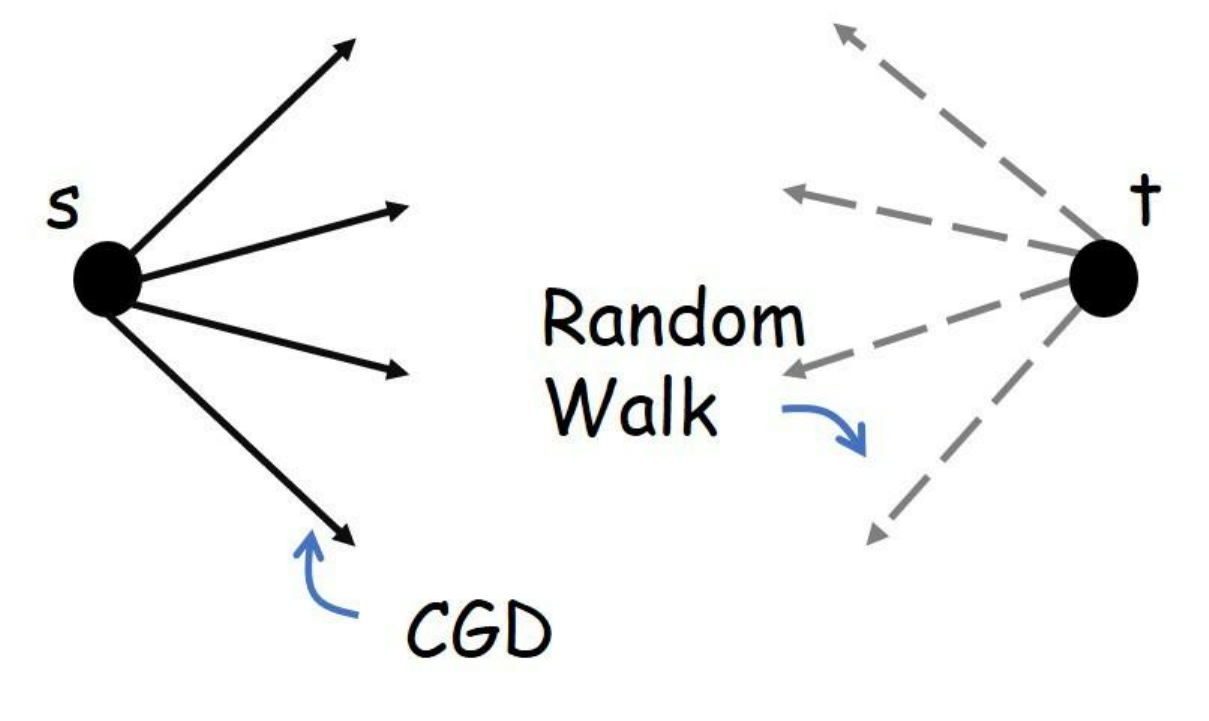} 
    }
    \subfigure[Our index-based algorithm for computing $\mathbf{p}_{L,u}$]{
    \includegraphics[scale=0.30]{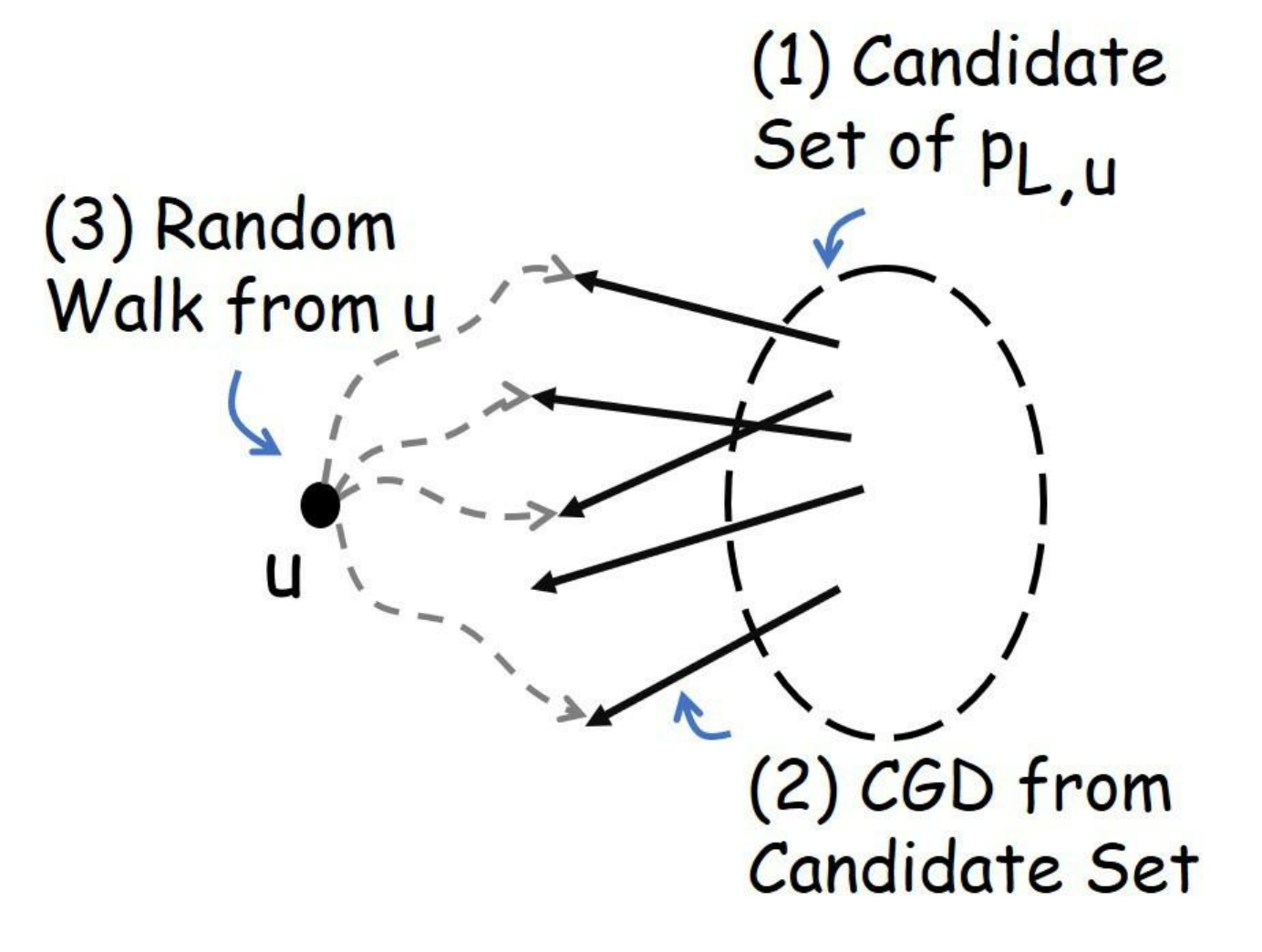} \vspace{-0.3cm}
    }
    \caption{Illustration of our algorithms for ER computation}\vspace{-0.5cm}
    \label{fig:our_algo}
\end{figure}

Next, we describe our algorithm to find the approximation $\hat{\mathbf{p}}_{L,u}(v)$ of $\mathbf{p}_{L,u}(v)$. The pseudo-code is outlined in Algorithm \ref{algo:bidir_single_pair}. Algorithm \ref{algo:bidir_single_pair} consists of two stages:

\begin{itemize}
\item Stage I: we perform a deterministic search, called \push, to obtain a coarse approximation of the vector $\mathbf{p}_{L,u}$, denote as $\mathbf{q}_{L,u}$. We output the residuals $\mathbf{r}_{i,u}$ for every $i\in[L]$.

\item Stage II: we perform random walks to obtain the precise approximation of $\mathbf{p}_{L,s}(s)$ (resp., $\mathbf{p}_{L,s}(t),\mathbf{p}_{L,t}(s),\mathbf{p}_{L,t}(t)$), denote as $\hat{\mathbf{p}}_{L,s}(s)$ (resp., $\hat{\mathbf{p}}_{L,s}(t),\hat{\mathbf{p}}_{L,t}(s),\hat{\mathbf{p}}_{L,t}(t)$). Thus we obtain the $\epsilon$-approximation of $r_\mathcal{G}(s,t)$.
\end{itemize}

As we discussed in Section \ref{sec:intro}, the high-level idea of our Algorithm is to use the deterministic search to reduce the variance produced by random walks, see Fig. \ref{fig:our_algo} (a) as an illustration. More precisely, for Stage I we use a variant of the classical coordinate gradient descent algorithm to approximate $\mathbf{p}_{L,u}$. We denote $\mathbf{x}_{l,u}=\frac{1}{2}(\frac{1}{2}\mathbf{I}+\frac{1}{2}\mathbf{P})^l\mathbf{e}_u$, therefore $\mathbf{p}_{L,u}=\frac{1}{2}\sum_{l=0}^{L}{\left(\frac{1}{2}\mathbf{I}+\frac{1}{2}\mathbf{P}\right)^l}\mathbf{e}_u=\sum_{l=0}^{L}{\mathbf{x}_{l,u}}$. Notice that $\mathbf{x}_{l+1,u}=(\frac{1}{2}\mathbf{I}+\frac{1}{2}\mathbf{P})\mathbf{x}_{l,u}$ for each $l< L$, so computing $\mathbf{p}_{L,u}$ is equivalent to solving the following linear system:
\begin{align*} \small
  \begin{bmatrix}
       \begin{array}{ccccc}
        \mathbf{I} &  & &  & \\
        -\frac{1}{2}\mathbf{I}-\frac{1}{2}\mathbf{P} &  \mathbf{I}  &  & & \\
         & \ddots  &  \ddots & & \\
         & & &\mathbf{I}  &\\
         & & & -\frac{1}{2}\mathbf{I}-\frac{1}{2}\mathbf{P} &   \mathbf{I} 
        \end{array}
        \end{bmatrix} 
\begin{bmatrix}
    \begin{array}{c}
         \mathbf{x}_{0,u}  \\
         \mathbf{x}_{1,u} \\
         \mathbf{x}_{2,u} \\
         \vdots \\
         \mathbf{x}_{L,u}
    \end{array}
\end{bmatrix}
=
\begin{bmatrix}
    \begin{array}{c}
         \frac{1}{2}\mathbf{e}_u \\
         0 \\
         0 \\
         \vdots \\
         0
    \end{array}
\end{bmatrix}
\end{align*}

For this linear system, we compute the approximation $\hat{\mathbf{x}}_{l,u}$ for each $l\in [L]$ and thus obtain the approximation $\mathbf{q}_{L,u}=\sum_{l=0}^{L}{\hat{\mathbf{x}}_{l,u}}$. For a given source node $u$, the Coordinate Gradient Descent (\push) Algorithm outputs the approximation $\mathbf{q}_{L,u}$ and residuals $\mathbf{r}_{i,u}$ for every $i\in[L]$ (without explicitly outputs each approximation $\hat{\mathbf{x}}_{l,u}$). Initially, we define the residual vector $\mathbf{r}_{i,u}=0$ for $i\in [1,...,L]$ and $\mathbf{r}_{0,u}=\frac{1}{2}\mathbf{e}_u$ and the approximation $\mathbf{q}_{L,u}=0$. Next, for each $i\in [0,...,L]$, for each $w\in \mathcal{V}$ with $\mathbf{r}_{i,u}(w)>\frac{r_{max}}{L^2}d_w$, we perform the following three steps (Line 4-11 of Algorithm \ref{algo:push}): (i) add the score $\mathbf{q}_{L,u}(w)$ by $\mathbf{r}_{i,u}(w)$; (ii) uniformly distribute $\frac{1}{2}\mathbf{r}_{i,u}(w)$ to the neighbor of $w$, namely $\mathbf{r}_{i+1,u}(v)\leftarrow \mathbf{r}_{i+1,u}(v)+\frac{1}{2d_w}\mathbf{r}_{i,u} (w)$ for $v\in \mathcal{N}(w)$ and $\mathbf{r}_{i+1,u}(w)\leftarrow \mathbf{r}_{i+1,u}(w)+\frac{1}{2}\mathbf{r}_{i,u}(w)$; (iii) set $\mathbf{r}_{i,u}(w)$ to $0$. When the iteration stops, Algorithm \ref{algo:push} outputs $\mathbf{q}_{L,u}$ as the coarse approximation of $\mathbf{p}_{L,u}$, and $L$ residual vectors $\mathbf{r}_{i,u}$ for each $i\in [L]$. For Stage II, we sample $n_r$ lazy random walks from the source node $s$, each random walk has length $L$ to refine the approximation (Line 6 of Algorithm \ref{algo:bidir_single_pair}). Specifically, a lazy random walk of length $L$ from $s$ denotes a sequence $(v_0=s,v_1,...,v_L)$ such that for every $1\leq k\leq L$, $v_{k}$ is a random neighbor of $v_{k-1}$ with probability $\frac{1}{2}$, and $v_{k+1}=v_k$ with probability $\frac{1}{2}$. For each step $k\leq L$ when the lazy random walk arrives at a node $w=v_k$, we update $\hat{\mathbf{p}}_{L,t}(s)\leftarrow \hat{\mathbf{p}}_{L,t}(s)+\sum_{i=0}^{L-k}{\frac{\mathbf{r}_{i,t}(w)}{d_w}}$ and $\hat{\mathbf{p}}_{L,s}(s)\leftarrow \hat{\mathbf{p}}_{L,s}(s)+\sum_{i=0}^{L-k}{\frac{\mathbf{r}_{i,s}(w)}{d_w}}$. We perform  similar steps for the sink node $t$ (Line 7 of Algorithm \ref{algo:bidir_single_pair}).

\begin{algorithm}[t!]
\small
	\SetAlgoLined
    \renewcommand{\algorithmicrequire}{\textbf{Input:}}
	\renewcommand{\algorithmicensure}{\textbf{Output:}}
    \caption{Our algorithm for single pair ER estimation}\label{algo:bidir_single_pair}
    \begin{algorithmic}[1]
	\REQUIRE $\mathcal{G},s,t, L,\epsilon$
    \STATE $d\leftarrow \min\{d_s,d_t\}$; $r_{max}\leftarrow \epsilon/\sqrt{d}$\; \hfill $\triangleright$ Stage I: deterministic search
	\STATE $\mathbf{q}_{L,s},\mathbf{r}_{i,s}$ for $i\in[L]\leftarrow$ \push ($\mathcal{G},s, L,r_{max}$) 
    \STATE $\mathbf{q}_{L,t},\mathbf{r}_{i,t}$ for $i\in[L]\leftarrow$ \push ($\mathcal{G},t, L,r_{max}$) 
    \STATE $n_r\leftarrow \frac{4L^2}{\epsilon^2}r_{max} d \log n$\; \hfill $\triangleright$ Stage II: random walk sampling
    \FOR{$j$ from $1$ to $n_r$}
    \STATE Perform lazy random walk of length $L$ start from $s$. For each step $k$ with $k\leq L$ it arrive at a node $w$, we update $\hat{\mathbf{p}}_{L,t}(s)\leftarrow \hat{\mathbf{p}}_{L,t}(s)+\frac{d_s}{n_r}\sum_{i=0}^{L-k}{\frac{\mathbf{r}_{i,t}(w)}{d_w}}$ and $\hat{\mathbf{p}}_{L,s}(s)\leftarrow \hat{\mathbf{p}}_{L,s}(s)+\frac{d_s}{n_r}\sum_{i=0}^{L-k}{\frac{\mathbf{r}_{i,s}(w)}{d_w}}$\;
    \STATE Perform lazy random walk of length $L$ start from $t$. For each step $k$ with $k\leq L$ it arrive at a node $w$, we update $\hat{\mathbf{p}}_{L,s}(t)\leftarrow \hat{\mathbf{p}}_{L,s}(t)+\frac{d_t}{n_r}\sum_{i=0}^{L-k}{\frac{\mathbf{r}_{i,s}(w)}{d_w}}$ and $\hat{\mathbf{p}}_{L,t}(t)\leftarrow \hat{\mathbf{p}}_{L,t}(t)+\frac{d_t}{n_r}\sum_{i=0}^{L-k}{\frac{\mathbf{r}_{i,t}(w)}{d_w}}$\;
    \ENDFOR
    \ENSURE $\hat{r}_\mathcal{G}(s,t)=\frac{\hat{\mathbf{p}}_{L,s}(s)}{d_s}-\frac{\hat{\mathbf{p}}_{L,s}(t)}{d_t}-\frac{\hat{\mathbf{p}}_{L,t}(s)}{d_s}+\frac{\hat{\mathbf{p}}_{L,t}(t)}{d_t}$ as the approximation of $r_\mathcal{G}(s,t)$
    \end{algorithmic}
\end{algorithm}

\begin{algorithm}[t!]
\small
	\SetAlgoLined
    \caption{Coordinate Gradient Descent (\push)}\label{algo:push}
    \renewcommand{\algorithmicrequire}{\textbf{Input:}}
	\renewcommand{\algorithmicensure}{\textbf{Output:}}
    \begin{algorithmic}[1]
	\REQUIRE $\mathcal{G},u, L,r_{max}$
   \STATE The approximation $\mathbf{q}_{L,u}(w)=0$ for each $w\in \mathcal{V}$\;
    \STATE $\mathbf{r}_{i,u}(w)=0$ for each $w\in \mathcal{V}$, $i\in [L+1]$, and $\mathbf{r}_{0,u}=\frac{1}{2}\mathbf{e}_s$\;
	\FOR{$i=0,1,2,...,L$}
        \WHILE{$\exists w\in \mathcal{V}$ such that $\mathbf{r}_{i,u}(w)>\frac{r_{max}}{L^2}d_w$}
            \STATE $\mathbf{q}_{L,u}(w)\leftarrow \mathbf{q}_{L,u}(w)+\mathbf{r}_{i,u}(w)$\;
            \FOR {$v\in \mathcal{N}(w)$}
               \STATE $\mathbf{r}_{i+1,u}(v)\leftarrow \mathbf{r}_{i+1,u}(v)+\frac{1}{2d_w}\mathbf{r}_{i,u} (w)$\;
            \ENDFOR
             \STATE $\mathbf{r}_{i+1,u}(w)\leftarrow \mathbf{r}_{i+1,u}(w)+\frac{1}{2}\mathbf{r}_{i,u}(w)$\;
            \STATE $\mathbf{r}_{i,u}(w)\leftarrow 0$\;
        \ENDWHILE
    \ENDFOR
    \ENSURE $\mathbf{q}_{L,u}$ as the coarse approximation of $\mathbf{p}_{L,u}$, residuals $\mathbf{r}_{i,u}$ for each $i\in [L]$
	\end{algorithmic}
\end{algorithm}

\section{Lower Bound for Online Single-Pair ER Computation}\label{sec:lower_bound}

In this section, we prove a lower bound for online single-pair ER computation on graphs. Below, we first describe our construction. Given the parameter $\epsilon\in(0,1)$, we consider two vertex sets $S_1$ and $S_2$, with $|S_1|=n_1$ and $|S_2|=n_2$. We define a $d$-regular expander graph on $S_1$. Note that this can be easily constructed because a random $d$-regular graph is an expander graph with high probability when $d\geq \log^2 n_1$, by ~\cite{friedman2003proof}. On $S_2$, we define an empty graph with only $n_2$ isolated vertices and $n_2\ll n_1$. Then, we add an additional (source) node $s$ with $d_s$ neighbors, such that $d_s\gg 1/\epsilon$. We randomly choose $x$ vertices in $S_1$, denote as $ \mathcal{N}_1=\{v_1,...,v_x\}\subset S_1$, and we randomly choose $(d_s-x)$ vertices in $S_2$, denote as $\mathcal{N}_2=\{v_{x+1},...,v_{d_s}\}\subset S_2$. We set the neighbors of $s$ to be $\mathcal{N}(s)=\{v_1,...,v_{d_s}\}=\mathcal{N}_1\cup \mathcal{N}_2$. For graph $\mathcal{G}_1$ we directly choose $x=d_s$, and for graph $\mathcal{G}_2$ we choose $x=(1-\epsilon)d_s$ for the given parameter $\epsilon\in (0,1)$. Finally, we add a (sink) node $t$ that connect to all vertices in $S_1$ and $ S_2$. 
We choose $d$ such that $d\geq \max \{d_s \log n_1,\log^2 n_1\}$ and $n_1\geq 2d^3$. See Fig. \ref{fig:lower_bound} (a),(b) as an illustration. Clearly, by our construction, any subset of $\mathcal{G}_1$ and $\mathcal{G}_2$ has $\Omega(1)$ expansion, so they are expanders. Then, we prove the following Theorem, which immediately implies Theorem \ref{thm:lower_bound}.\hypertarget{proof_of_theorem:lower_bound}{}

\begin{figure}
    \centering
    \subfigure[graph $\mathcal{G}_1$]{
    \includegraphics[scale=0.3]{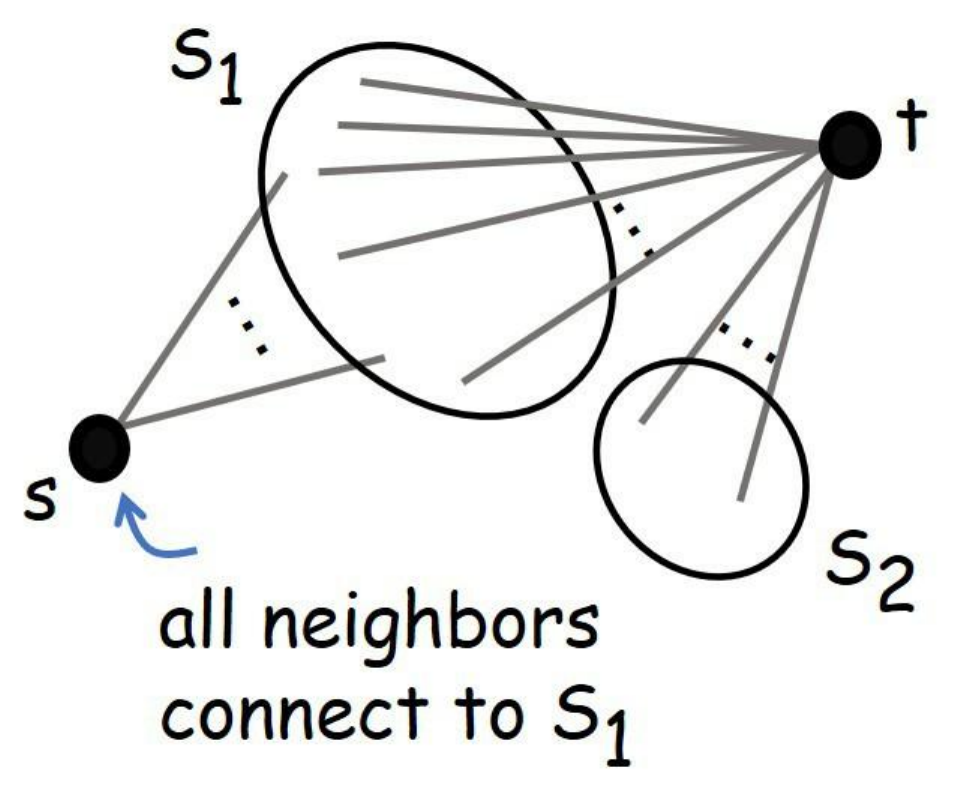} \vspace{-0.3cm}
    }
    \subfigure[graph $\mathcal{G}_2$]{
    \includegraphics[scale=0.3]{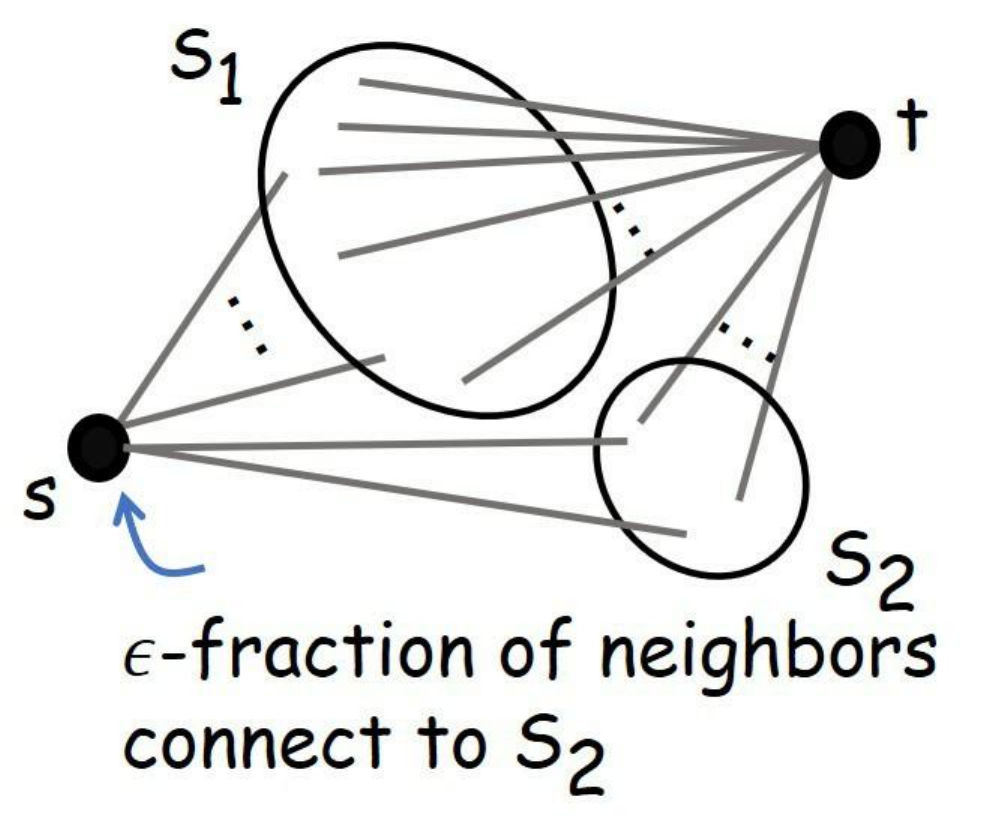} \vspace{-0.3cm}
    }
    \subfigure[parallel resistor illustration]{
    \includegraphics[scale=0.25]{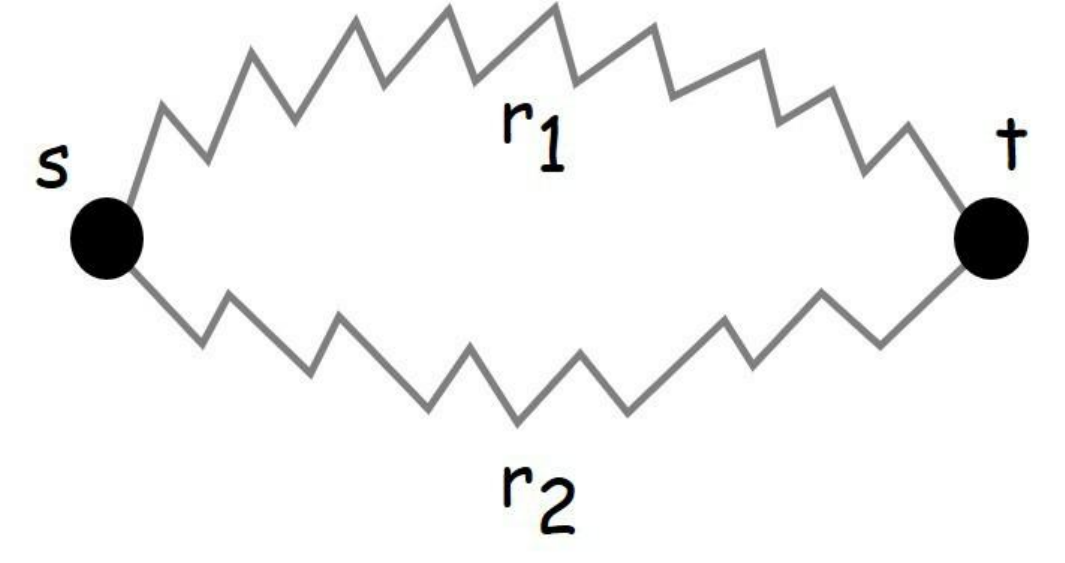}
    }
    \caption{illustration of the construction of our lower bound}\vspace{-0.5cm}
    \label{fig:lower_bound}
\end{figure}

\begin{theorem}
\label{thm:lower_bound_specific}
    $|r_{\mathcal{G}_1}(s,t)-r_{\mathcal{G}_2}(s,t)|> \frac{\epsilon}{4} r_{\mathcal{G}_1}(s,t)$, but any randomized local algorithm requires $\Omega(\frac{1}{\epsilon})$ time to distinguish $\mathcal{G}_1$ and $\mathcal{G}_2$ with probability $\geq 2/3$.
\end{theorem}

First, to get intuitions for why $|r_{\mathcal{G}_1}(s,t)-r_{\mathcal{G}_2}(s,t)|> \frac{\epsilon}{4} r_{\mathcal{G}_1}(s,t)$, we use the representation of parallel resistors, see Fig. \ref{fig:lower_bound} (c) as an illustration. By our construction, we connect $x$ neighbors of $s$ to $S_1$ and connect $(d_s-x)$ neighbors of $s$ to $S_2$. Then, we connect the two parts of the resistor $S_1,S_2$ in parallel. For $\mathcal{G}_1$ all the neighbors of $s$ is in $S_1$; for $\mathcal{G}_2$ we choose $x=(1-\epsilon)d_s$, thereby there are $(1-\epsilon)$ fraction of the neighbors of $s$ in $S_1$ and $\epsilon$ fraction of the neighbors in $S_2$. We prove the $s,t$ resistance from $S_1$ is $\frac{1}{x}+o(\frac{1}{x})$; while the $s,t$ resistance from $S_2$ is $\frac{2}{d_s-x}$. As a result, the ER value $r_{\mathcal{G}_1}(s,t)$ and $r_{\mathcal{G}_2}(s,t)$ will differ by $\frac{\epsilon}{4}$ in terms of relative error. To reach this end, we first prove the following Lemma.


\begin{lemma}[\downlink{proof_of_lemma:parallel_resistance}]\label{lem:parallel_resistance}
    $r_{S_1}(s,t)=\frac{1}{x}+O(\frac{1}{dx}\log n_1 )=\frac{1}{x}+O(\frac{1}{d_s x} )$, and $r_{S_2}(s,t)=\frac{2}{d_s-x}$. Where $r_{S_1}(s,t)$ denotes the $s,t$ resistance only consider the edges between $S_1\cup \{s,t\}$ and $r_{S_2}(s,t)$ denotes the $s,t$ resistance only consider the edges between $S_2\cup \{s,t\}$.
\end{lemma}

Armed with Lemma \ref{lem:parallel_resistance} and the parallel resistance formula, we prove that $r_{\mathcal{G}_1}(s,t)$ and $r_{\mathcal{G}_2}(s,t)$ differs by $\frac{\epsilon}{4}$ in terms of relative error.

\begin{lemma}[\downlink{proof_of_lemma:G1G2_differ}]\label{lem:G1G2_differ}
    $|r_{\mathcal{G}_1}(s,t)-r_{\mathcal{G}_2}(s,t)|> \frac{\epsilon}{4} r_{\mathcal{G}_1}(s,t).$
\end{lemma}

Finally, we prove that any local algorithms to distinguish $\mathcal{G}_1$ and $\mathcal{G}_2$ requires at least $\Omega(\frac{1}{\epsilon})$ queries. 

\begin{lemma}[\downlink{proof_of_lemma:G1G2_distinguish}]\label{lem:G1G2_distinguish}
    Any local algorithm that distinguish $\mathcal{G}_1$ and $\mathcal{G}_2$ with probability $\geq 2/3$ requires at least $\Omega(\frac{1}{\epsilon})$ queries.
\end{lemma}

Combining Lemma \ref{lem:G1G2_differ} and Lemma \ref{lem:G1G2_distinguish}, we immediately prove Theorem \ref{thm:lower_bound_specific}, which implies Theorem \ref{thm:lower_bound}.

\section{Index-based Algorithms for ER Computation}

In this section, we present a more efficient index-based ER computation algorithm and prove Theorem \ref{thm:ER_sketch}. Note that by Theorem \ref{thm:bidir_complexity}, if we directly use Algorithm \ref{algo:bidir_single_pair} to compute the all pair ER values, the time complexity is bounded by:
\begin{align*} \small
    \tilde{O}\left(\sum_{u\in \mathcal{V}}{\sum_{v\in \mathcal{V},d_v\leq d_u}{\frac{\sqrt{d_v}}{\epsilon}}}\right)
    &\leq \tilde{O}\left(\sum_{u\in \mathcal{V}}{\sum_{v\in \mathcal{V}}{\frac{\sqrt{d_v}}{\epsilon}}}\right)\\
   (Cauchy-Schwarz) &\leq \tilde{O}\left(\frac{1}{\epsilon}\sum_{u\in \mathcal{V}}{\left(\left(\sum_{v\in \mathcal{V}}{1^2}\right)^{1/2}\left(\sum_{v\in \mathcal{V}}{(\sqrt{d_v})^2}\right)^{1/2}\right)}\right)\\
    &=\tilde{O}\left(\frac{1}{\epsilon}\sum_{u\in \mathcal{V}}{\sqrt{nm}}\right)=\tilde{O}\left(\frac{1}{\epsilon}n\sqrt{nm}\right).
\end{align*}

In this section, we will prove that $\tilde{O}(\sqrt{nm}/{\epsilon})$ processing time is enough for the ER sketch algorithms. Recall that by Eq. (\ref{equ:ER_p_expression}), given any vertex pair $(s,t)$, one can approximate $r_\mathcal{G}(s,t)$ by the linear combination of $\mathbf{p}_{L,s}(s),\mathbf{p}_{L,s}(t),\mathbf{p}_{L,t}(s)$ and $\mathbf{p}_{L,t}(t)$. Therefore, if one can approximates these $n$ vectors $\mathbf{p}_{L,u}$ for all $u\in \mathcal{V}$, one can approximate $r_\mathcal{G}(s,t)$ for any given vertex pair $(s,t)$. We prove the following Theorem.

\begin{theorem}[\downlink{proof_of_theorem:ER_sketch_specific}]\label{thm:ER_sketch_specific}
There exists an algorithm (i.e. Algorithm \ref{algo:bidir_sketch}) that outputs the approximation $\hat{\mathbf{p}}_{L,u}$ for every $u\in \mathcal{V}$, such that $|\hat{\mathbf{p}}_{L,u}(v)-\mathbf{p}_{L,u}(v)|\leq \epsilon$ for every $ v$ with $d_v\leq d_u$. In addition, the time complexity of this algorithm is $\tilde{O}(\sqrt{nm}/{\epsilon})$, the space complexity is $\tilde{O}(n/{\epsilon})$.
\end{theorem}

The key observation here is that we can approximate the vector $\mathbf{p}_{L,u}$ within almost the same time as just approximate $\mathbf{p}_{L,u}(v)$ for a vertex $v$. The pseudo-code of our ER sketch algorithm is outlined in Algorithm \ref{algo:bidir_sketch}. To compute the approximation vector $\hat{\mathbf{p}}_{L,u}$, Algorithm \ref{algo:bidir_sketch} consists of three stages, see Fig. \ref{fig:our_algo} (b) as an illustration.

\begin{itemize}
\item Stage I: perform $n_r^{(1)}$ random walks to find a candidate set, denote as $S'_{\epsilon,u}=\{w\in \mathcal{V}:\mathbf{p}'_{L,u}(w)\geq \epsilon/2\}$;

\item Stage II: perform \push operations from each $v\in S'_{\epsilon,u}$;

\item Stage III: perform random walks start from $u$ to derive the final approximation $\hat{\mathbf{p}}_{L,u}$.
\end{itemize}

Specifically, our algorithm is implemented as follows. First, we perform a traversal for all $u\in \mathcal{V}$. For each $u$, we compute the approximation $\hat{\mathbf{p}}_{L,u}$. For Stage I, we perform $n_r^{(1)}$ random walks from the source node $u$ to compute a very coarse approximation $\mathbf{p}'_{L,u}$ (Lines 2-4 in Algorithm \ref{algo:bidir_sketch}). We define the candidate set $S'_{\epsilon,u}=\{w\in \mathcal{V}:\mathbf{p}'_{L,u}(w)> \epsilon/2\}$ and next we focus on the approximation on the candidate set $S'_{\epsilon,u}$ (the size of the candidate set $S'_{\epsilon,u}$ is only $\tilde{O}(\frac{1}{\epsilon})$, independent of $n$). For Stage II, we perform \push operations from each $v\in S'_{\epsilon,u}$ with the threshold $\frac{r_{max}}{2 \mathbf{p}'_{L,u}(v)}$ (Lines 5-8 in Algorithm \ref{algo:bidir_sketch}). For Stage III, we perform $n_r^{(3)}$ random walks from the source node $u$ (Lines 9-10 in Algorithm \ref{algo:bidir_sketch}). The key mechanism here is that the random walks can share the computation of \push for different $\mathbf{p}_{L,u}(v_1)$ and $\mathbf{p}_{L,u}(v_2)$ with $v_1\neq v_2$. Therefore, to approximate the vector $\mathbf{p}_{L,u}$ we only need to perform random walks from the source node $u$ without the traversal for $w\in \mathcal{V}$. Finally, we update the result and get the final approximation $\hat{\mathbf{p}}_{L,u}(v)$ for every $v\in S'_{\epsilon,u}$ with $d_v\leq d_u$ (Lines 11-13 in Algorithm \ref{algo:bidir_sketch}). For this step, the operation time is at most the same as Stage II, this enables the local computation of the vector $\mathbf{p}_{L,u}$.

\begin{algorithm}[t!]
\small
\SetAlgoLined
    \renewcommand{\algorithmicrequire}{\textbf{Input:}}
	\renewcommand{\algorithmicensure}{\textbf{Output:}}
    	\caption{Our algorithm for building the index}\label{algo:bidir_sketch}
    \begin{algorithmic}[1]
	\REQUIRE $\mathcal{G},s,t, L,\epsilon$
    \FOR{$u\in \mathcal{V}$}
   \STATE $n_r^{(1)}\leftarrow \frac{ L^2\log n}{\epsilon}$, $\mathbf{p}'_{L,u}\leftarrow 0$  \hfill $\triangleright$ Stage I: find the $\epsilon$-contributing set\;
   \STATE Perform $n_r^{(1)}$ lazy random walks of length $L$ start from $u$. For each lazy random walk, for each step $k$ with $k\leq L$ it arrive at a node $w$, we update $\mathbf{p}'_{L,u}(w)\leftarrow \mathbf{p}'_{L,u}(w)+\frac{1}{2n_r^{(1)}}$\;
    \STATE Define $S'_{\epsilon,u}=\{w\in \mathcal{V}:\mathbf{p}'_{L,u}(w)> \epsilon/2\}$ as the candidate set\;
   \STATE $r_{max}\leftarrow L\epsilon\sqrt{\frac{n}{m}}$ \hfill $\triangleright$ Stage II: \push from candidate set $S_{\epsilon,u}'$\;
    \FOR{$v\in S'_{\epsilon,u}$}
	  \STATE $\mathbf{q}_{L,v},\mathbf{r}_{i,v}$ for $i\in[L]\leftarrow$ \push ($\mathcal{G},w, L,\frac{r_{max}}{2 \mathbf{p}'_{L,u}(v)}$)\;
    \ENDFOR
    \STATE $n_r^{(3)}\leftarrow \frac{L}{\epsilon^2}r_{max}d_u\log  n$ \hfill $\triangleright$ Stage III: random walk sampling from $u$\;
     \STATE Perform $n_r^{(3)}$ lazy random walks of length $L$ start from $u$, denote $N^{(3)}_{u,w,k}$ as the number of walks arrive at $w$ at step $k$ with $k\leq L$.
    \FOR{$v\in S'_{\epsilon,u}$ with $d_v\leq d_u$}
        \STATE $\hat{\mathbf{p}}_{L,u}(v)\leftarrow \frac{\mathbf{q}_{L,v}(u)d_v}{d_u}+\frac{d_v}{n_r^{(3)}}\sum_{k=0}^{L}{\sum_{w\in\mathcal{V}}{\left(\sum_{i=0}^{L-k}{\frac{\mathbf{r}_{i,u}(w)}{d_w}}\right)N^{(3)}_{u,w,k}}}$\;
    \ENDFOR
    \ENDFOR
    \ENSURE $\hat{\mathbf{p}}_{L,u}$ for every $u\in \mathcal{V}$
\end{algorithmic}
\end{algorithm}

\begin{algorithm}[t!]
\small
	\SetAlgoLined
    \caption{Our query algorithm for $s,t$-ER value}\label{algo:sketch_query}
    \renewcommand{\algorithmicrequire}{\textbf{Input:}}
	\renewcommand{\algorithmicensure}{\textbf{Output:}}
    \begin{algorithmic}[1]
	\REQUIRE $s,t$, $\hat{\mathbf{p}}_{L,u}$ for all $u\in\mathcal{V}$
    \IF{$d_s\leq d_t$}
    \STATE $\hat{r}_\mathcal{G}(s,t)=\frac{\hat{\mathbf{p}}_{L,s}(s)}{d_s}-2\frac{\hat{\mathbf{p}}_{L,t}(s)}{d_s}+\frac{\hat{\mathbf{p}}_{L,t}(t)}{d_t}$\;
    \ELSE
   \STATE $\hat{r}_\mathcal{G}(s,t)=\frac{\hat{\mathbf{p}}_{L,s}(s)}{d_s}-2\frac{\hat{\mathbf{p}}_{L,s}(t)}{d_t}+\frac{\hat{\mathbf{p}}_{L,t}(t)}{d_t}$\;
    \ENDIF
    \ENSURE $\hat{r}_\mathcal{G}(s,t)$
    \end{algorithmic}
\end{algorithm}

\section{Conclusion}
In this paper, we propose several new algorithms for approximating Effective Resistance (ER) with reduced dependency on the error parameter $\epsilon$. For online ER computation algorithms, we integrate deterministic search with random walk sampling to develop an $\tilde{O}(\sqrt{d}/{\epsilon})$-time algorithm that $\epsilon$-approximates the single-pair ER value in expander graphs. Additionally, we establish that $\Omega(1/\epsilon)$ represents the lower bound for local ER approximation, even in the case of expander graphs. For index-based ER computation algorithms, we extend our techniques and propose an ER sketch algorithm that advances the state-of-the-art. For the open problems, beyond the natural challenge of improving both upper and lower bounds, it would particularly be interested in exploring the relationship between $\kappa(\mathcal{L})$ and $\epsilon$ in the context of the lower bound for local ER computation, so as to unify the analysis for both expander and non-expander graphs.

\acks{This work is supported by the Funds of the National Natural Science Foundation of China (NFSC) No.U2241211 and U24A20255. Rong-Hua Li is the corresponding author of this paper.}


\bibliography{ref}

\appendix

\section{Guideline of the Appendix} 
In Appendix \ref{sec:parallel_resistance}, we provide some explainations for the parallel resistance formula. In Appendix \ref{sec:analysis_single_pair}, we provide the theoretical analysis of Algorithm \ref{algo:bidir_single_pair} (i.e. our query processing algorithm for ER computation of a single pair) and prove Theorem \ref{thm:bidir_complexity}. In Appendix \ref{sec:analysis_sketch}, we provide the theoretical analysis of Algorithm \ref{algo:bidir_sketch} (i.e. our index based algorithm for ER computation) and prove Theorem \ref{thm:ER_sketch}. In Appendix \ref{sec:other_proofs}, we provide the other omitting proofs of this paper, including the omitting proofs in section \ref{sec:basic_repre} and section \ref{sec:lower_bound}.

\section{Parallel Resistance Formula}\label{sec:parallel_resistance}

For the proof of our lower bound, we provide some details of the parallel resistance formula. We begin by the definition of Schur complement.

\begin{definition}
    Given a graph $\mathcal{G}=(\mathcal{V},\mathcal{E})$, The Schur complement of a subset $S\subset \mathcal{V}$ is a weighted graph with the corresponding Laplacian matrix $\mathbf{SC}(\mathcal{G},S)=\mathbf{L}_{[S,S]}-\mathbf{L}_{[S,\bar{S}]}\mathbf{L}_{[\bar{S},\bar{S}]}^{-1}\mathbf{L}_{[\bar{S},S]}$, where $\bar{S}=\mathcal{V}- S$ and $\mathbf{L}_{[S,S]}$ denotes the submatrix of $\mathbf{L}$ with row and column indexed by $S$.
\end{definition}

There is a well known result that the $s,t$-ER value on graph $\mathcal{G}$ is equal to the $s,t$-ER value on the Schur complement.

\begin{theorem}
    For a subset $S\subset \mathcal{V}$ such that $s,t\in S$, we have $r_\mathcal{G}(s,t)=(\mathbf{e}_s-\mathbf{e}_t)^T\mathbf{L}^\dagger (\mathbf{e}_s-\mathbf{e}_t)=(\mathbf{e}_s-\mathbf{e}_t)^T\mathbf{SC}(\mathcal{G},S)^\dagger (\mathbf{e}_s-\mathbf{e}_t)$.
\end{theorem}

Now we consider a graph $\mathcal{G}=(\mathcal{V},\mathcal{E})$, such that $\mathcal{V}=\{s,t\}\sqcup S_1\sqcup S_2$, and there are no edges between $S_1$ and $S_2$, i.e. $\mathcal{E}(S_1,S_2)=\emptyset$, and $(s,t)\notin \mathcal{E}$. We denote $\mathcal{G}_{S_1}$ as the induced subgraph on $\{s,t\}\sqcup S_1$, and $\mathcal{G}_{S_2}$ as the induced subgraph on $\{s,t\}\sqcup S_2$. We denote $r_{\mathcal{G}}(s,t)$ as the $s,t$-ER value on $\mathcal{G}$, $r_{S_1}(s,t)$ as the $s,t$-ER value on subgraph $\mathcal{G}_{S_1}$, and $r_{S_2}(s,t)$ as the $s,t$-ER value on subgraph $\mathcal{G}_{S_2}$, respectively. Then we provide the following parallel resistance formula.

\begin{theorem}\label{thm:parallel_resistance}
    $\frac{1}{r_\mathcal{G}(s,t)}=\frac{1}{r_{S_1}(s,t)}+\frac{1}{r_{S_2}(s,t)}$.
\end{theorem}

\begin{proof}
    We prove this result by linear algebra argument instead of physics intuition. First, we consider the Schur complement of $\mathcal{G}_{S_1}$ on $\{s,t\}$, $\mathbf{SC}(\mathcal{G}_{S_1},\{s,t\})$. We notice that this is a graph with single (weighted) edge between $s,t$. So we denote $\mathbf{SC}(\mathcal{G}_{S_1},\{s,t\})=w_{s,t} (\mathbf{e}_s-\mathbf{e}_t)(\mathbf{e}_s-\mathbf{e}_t)^T$. Similarly, we denote $\mathbf{SC}(\mathcal{G}_{S_2},\{s,t\})=w_{s,t}' (\mathbf{e}_s-\mathbf{e}_t)(\mathbf{e}_s-\mathbf{e}_t)^T$. Clearly, by definition $w_{s,t}=1/r_{S_1}(s,t)$ and $w_{s,t}'=1/r_{S_2}(s,t)$. Next, we consider the Schur complement of $\mathcal{G}$ on $\{s,t\}$, $\mathbf{SC}(\mathcal{G},\{s,t\})$. By $\mathcal{E}(S_1,S_2)=\emptyset$ and the fact that taking the Schur complement is equivalent to the partial Gaussian elimination procedure ~\cite{kyng2016approximate}, $\mathbf{SC}(\mathcal{G},\{s,t\})$ is a graph with single edge between $s,t$ with edge weight $w_{s,t}+w_{s,t}'$ (because the process of eliminating $S_1$ and eliminating $S_2$ are independent). By definition $w_{s,t}+w_{s,t}'=1/r_{\mathcal{G}}(s,t)$. This finishes the proof.
\end{proof}

By the same argument, we can generalize the parallel resistance formula to the $k$ parallel resistance case. That is, when we consider $\mathcal{V}=\{s,t\}\sqcup S_1\sqcup S_2\sqcup...\sqcup S_k$ with $\mathcal{E}(S_i,S_j)=\emptyset$, $i,j\in [k]$. Then $\frac{1}{r_\mathcal{G}(s,t)}=\sum_{i=1}^{k}{\frac{1}{r_{S_i}(s,t)}}+\mathbb{I}_{(s,t)\in \mathcal{E}}$, where $\mathbb{I}_{(s,t)\in \mathcal{E}}$ is the indicator function. Since the proof is almost same as Theorem \ref{thm:parallel_resistance}, we omit it here.

\section{Theoretical analysis of Algorithm \ref{algo:bidir_single_pair}}\label{sec:analysis_single_pair}

In this section, we prove the theoretical guarantee output by Algorithm \ref{algo:bidir_single_pair}. To this end, we begin by analyzing the output guarantee by \push (Algorithm \ref{algo:push}). The following Lemma shows the relationship between the approximation vector $\mathbf{q}_{L,u}$ output by \push and the accurate vector $\mathbf{p}_{L,u}$.

\begin{lemma}
\label{lem:invariant}
    The following equation holds at any iteration of Algorithm \ref{algo:push}:
    \begin{equation}\label{equ:invariant}
    \begin{aligned}
        \mathbf{p}_{L,u}=\mathbf{q}_{L,u}+\sum_{i=0}^{L}{\sum_{k=i}^{L}{\left(\frac{1}{2}\mathbf{I}+\frac{1}{2}\mathbf{P}\right)^{k-i}\mathbf{r}_{i,u}}}.
    \end{aligned}
    \end{equation}
\end{lemma}
\stitle{Proof of Lemma \ref{lem:invariant}.}\hypertarget{proof_of_lemma:invariant}{}
    The proof is by induction. For the initial step, $\mathbf{r}_{0,u}=\frac{1}{2}\mathbf{e}_s$ and $\mathbf{r}_{i,u}=0$ for each $i\in [1,2,...,L]$. Therefore the right hand side of Equ. (\ref{equ:invariant}) equals $\sum_{i=0}^{L}{\left(\frac{1}{2}\mathbf{I}+\frac{1}{2}\mathbf{P}\right)^{k-i}\mathbf{r}_{0,u}}=\mathbf{p}_{L,u}$. Next we assume that after $j$ \push operations, the approximation vector $\mathbf{q}_{L,u}$ and the residuals $\mathbf{r}_{0,u},\mathbf{r}_{1,u},...,\mathbf{r}_{L,u}$ satisfy Equ. (\ref{equ:invariant}). For the $(j+1)^{th}$ \push iteration, we consider the operation on $(w,i)$. That is, we perform following three operations: (i) $\mathbf{q}_{L,u}(w)\leftarrow \mathbf{q}_{L,u}(w)+\mathbf{r}_{i,u}(w)$; (ii) $\mathbf{r}_{i+1,u}(v)\leftarrow \mathbf{r}_{i+1,u}(v)+\frac{1}{2d_w}\mathbf{r}_{i,u} (w)$ for each $v\in \mathcal{N}(w)$ and $\mathbf{r}_{i+1,u}(w)\leftarrow \mathbf{r}_{i+1,u}(w)+\frac{1}{2}\mathbf{r}_{i,u}(w)$; (iii) $\mathbf{r}_{i,u}(w)\leftarrow 0$. We define the resulting vector $\hat{\mathbf{q}}_{L,u}$ and $\hat{\mathbf{r}}_{0,u},\hat{\mathbf{r}}_{1,u},...,\hat{\mathbf{r}}_{L,u}$ after this \push operation, and consider the difference $\Delta(w,i)$ of the right hand side of Eq. (\ref{equ:invariant}):
    \begin{align*}
        \Delta(w,i)&=\mathbf{q}_{L,u}+\sum_{i=0}^{L}{\sum_{k=i}^{L}{\left(\frac{1}{2}\mathbf{I}+\frac{1}{2}\mathbf{P}\right)^{k-i}\mathbf{r}_{i,u}}}
        -\hat{\mathbf{q}}_{L,u}+\sum_{i=0}^{L}{\sum_{k=i}^{L}{\left(\frac{1}{2}\mathbf{I}+\frac{1}{2}\mathbf{P}\right)^{k-i}\hat{\mathbf{r}}_{i,u}}}.\\
    \end{align*}
    We note the following three equations hold after the \push operation on $(w,i)$:
    \begin{align*}
    \mathbf{q}_{L,u}-\hat{\mathbf{q}}_{L,u}&=\mathbf{r}_{i,u}(w)\mathbf{e}_w;\\
    \mathbf{r}_{i,u}-\hat{\mathbf{r}}_{i,u}&=-\mathbf{r}_{i,u}(w)\mathbf{e}_w;\\
    \mathbf{r}_{i+1,u}-\hat{\mathbf{r}}_{i+1,u}&=\frac{1}{2}\mathbf{r}_{i,u}(w)\mathbf{e}_w+\sum_{v\in \mathcal{N}(w)}{\frac{1}{2d_w}\mathbf{r}_{i,u} (w)\mathbf{e}_v}\\
    &=\mathbf{r}_{i,u} (w)\left(\frac{1}{2}\mathbf{I}+\frac{1}{2}\mathbf{P}\right)\mathbf{e}_w.
    \end{align*}
    Therefore,
    \begin{align*}
        \Delta(w,i)&=\mathbf{r}_{i,u}(w)\mathbf{e}_w-\mathbf{r}_{i,u}(w)\sum_{k=i}^{L}{\left(\frac{1}{2}\mathbf{I}+\frac{1}{2}\mathbf{P}\right)^{k-i}\mathbf{e}_w}\\
        &+\mathbf{r}_{i,u}(w)\sum_{k=i+1}^{L}{\left(\frac{1}{2}\mathbf{I}+\frac{1}{2}\mathbf{P}\right)^{k-i}\left(\frac{1}{2}\mathbf{I}+\frac{1}{2}\mathbf{P}\right)\mathbf{e}_w}\\
        &=\mathbf{r}_{i,u}(w)-\mathbf{r}_{i,u}(w)=0.
    \end{align*}
    Thus Eq. (\ref{equ:invariant}) also holds after $(j+1)^{th}$ \push operation. By induction, Eq. (\ref{equ:invariant}) holds at any iteration of Algorithm \ref{algo:push}. This finishes the proof.\hfill$\blacksquare$\par

Next, we show that \push outputs the approximation vector $\mathbf{q}_{L,u}$ satisfies the certain error guarantee.

\begin{lemma}
\label{lem:guarantee_push}
   \push outputs the approximation $\mathbf{q}_{L,u}$ such that $\mathbf{p}_{L,u}(w)-r_{max} d_w \leq \mathbf{q}_{L,u}(w)\leq \mathbf{p}_{L,u}(w)$ for any $w\in \mathcal{V}$. Furthermore, the time complexity of \push is $O(\frac{L^3}{r_{max}})$.
\end{lemma}

\stitle{Proof of Lemma \ref{lem:guarantee_push}.}\hypertarget{proof_of_lemma:guarantee_push}{}
    By Lemma \ref{lem:invariant}, Eq. (\ref{equ:invariant}) holds at any iteration of \push. By our implementation, the residuals satisfy $\mathbf{r}_{i,u}(w)\leq \frac{r_{max}d_w}{L^2}$ for any $w\in \mathcal{V}$, $i\in[L]$. Therefore, we have:
    \begin{align*}
        \Vert \mathbf{D}^{-1}(\mathbf{p}_{L,u}-\mathbf{q}_{L,u})\Vert_\infty&=\left\Vert \mathbf{D}^{-1}\sum_{i=0}^{L}{\sum_{k=i}^{L}{\left(\frac{1}{2}\mathbf{I}+\frac{1}{2}\mathbf{P}\right)^{k-i}\mathbf{r}_{i,u}}}\right\Vert_\infty\\
        &\leq  \sum_{i=0}^{L}{\sum_{k=i}^{L}{\left\Vert\mathbf{D}^{-1}\left(\frac{1}{2}\mathbf{I}+\frac{1}{2}\mathbf{P}\right)^{k-i}\mathbf{r}_{i,u}\right\Vert_\infty}}\\
        &=\sum_{i=0}^{L}{\sum_{k=i}^{L}{\left\Vert\left(\frac{1}{2}\mathbf{I}+\frac{1}{2}\mathbf{P}^T\right)^{k-i}\mathbf{D}^{-1}\mathbf{r}_{i,u}\right\Vert_\infty}}\\
        &\leq \sum_{i=0}^{L}{\sum_{k=i}^{L}{\left\Vert\left(\frac{1}{2}\mathbf{I}+\frac{1}{2}\mathbf{P}^T\right)^{k-i}\right\Vert_\infty \left\Vert\mathbf{D}^{-1}\mathbf{r}_{i,u}\right\Vert_\infty}}.\\
    \end{align*}
    Next, we notice that $\left\Vert\left(\frac{1}{2}\mathbf{I}+\frac{1}{2}\mathbf{P}^T\right)^{k-i}\right\Vert_\infty =1$ since $\left(\frac{1}{2}\mathbf{I}+\frac{1}{2}\mathbf{P}^T\right)^{k-i}$ is a row-stochastic matrix and $\Vert\mathbf{D}^{-1}\mathbf{r}_{i,u}\Vert_\infty\leq \frac{r_{max}}{L^2}$. Therefore,
    \begin{align*}
        \Vert \mathbf{D}^{-1}(\mathbf{p}_{L,u}-\mathbf{q}_{L,u})\Vert_\infty&\leq \sum_{i=0}^{L}{\sum_{k=i}^{L}{\frac{r_{max}}{L^2}}}\leq r_{max}.
    \end{align*}
    Thus $\mathbf{p}_{L,u}(w)-r_{max} d_w \leq \mathbf{q}_{L,u}(w)\leq \mathbf{p}_{L,u}(w)+r_{max} d_w$ for any $w\in \mathcal{V}$. Furthermore, we notice that $\mathbf{r}_{i,u}(w)\geq 0$ for $\forall w\in \mathcal{V}$, $i\in[L]$. By Eq. (\ref{equ:invariant}), we obtain $\mathbf{q}_{L,u}(w)\leq \mathbf{p}_{L,u}(w)$ for $\forall w\in \mathcal{V}$. Putting it together, we have $\mathbf{p}_{L,u}(w)-r_{max} d_w \leq \mathbf{q}_{L,u}(w)\leq \mathbf{p}_{L,u}(w)$.

    For the time complexity of \push, we observe that at the beginning of each iteration (Line 4 in Algorithm \ref{algo:push}), $\Vert \mathbf{r}_{i,u}\Vert_1\leq 1$, since $\Vert \mathbf{r}_{i,u}\Vert_1$ is non-increasing with $i$. Next, we only invoke nodes $w\in\mathcal{V}$ with $\mathbf{r}_{i,u}(w)>\frac{r_{max}d_w}{L^2}$, for each $w$ the operation numbers is $O(d_w)$ (Line 5-9 in Algorithm \ref{algo:push}). Finally, we set $\mathbf{r}_{i,u}(w)\leftarrow 0$ after the operation on $w$ (Line 10 in Algorithm \ref{algo:push}). Thus, one unit operation will decrease $\Vert \mathbf{r}_{i,u}\Vert_1$ by at least $\frac{r_{max}}{L^2}$ in average. So the total operations of \push in each iteration (Line 4-11 in Algorithm \ref{algo:push}) can be bounded by $O(\frac{L^2}{r_{max}})$. Since Algorithm \ref{algo:push} has $L$ iterations, the total time complexity is bounded by $O(\frac{L^3}{r_{max}})$.\hfill$\blacksquare$\par

By Lemma \ref{lem:guarantee_push}, when setting $r_{max}=\frac{\epsilon}{d}$, the approximation guarantee by \push already satisfies Theorem \ref{thm:bidir_guarantee_specific} and thus this is an algorithm for $\epsilon$-approximate ER estimation. However, the time complexity of \push is $O(\frac{L^3}{r_{max}})=\tilde{O}(\frac{d}{\epsilon})$. Next we will prove that combining with random walk sampling (i.e. Stage II in algorithm \ref{algo:bidir_single_pair}), this time complexity can be reduced to $\tilde{O}(\frac{\sqrt{d}}{\epsilon})$ while maintaining the same error guarantee. To this end, we consider again the output guarantee of \push for a given node $v\in \mathcal{V}$:

\begin{align*}
        \mathbf{p}_{L,u}(v)=\mathbf{q}_{L,u}(v)+\sum_{i=0}^{L}{\sum_{k=i}^{L}{\left(\left(\frac{1}{2}\mathbf{I}+\frac{1}{2}\mathbf{P}\right)^{k-i}\mathbf{r}_{i,u}\right)(v)}}.
\end{align*}
We divide $d_v$ on both sides of the equation:
    \begin{equation}\label{equ:rw_representation}
    \begin{aligned}
        \frac{\mathbf{p}_{L,u}(v)}{d_v}&=\frac{\mathbf{q}_{L,u}(v)}{d_v}+\sum_{i=0}^{L}{\sum_{k=i}^{L}{\left(\left(\frac{1}{2}\mathbf{I}+\frac{1}{2}\mathbf{P}\right)^{k-i}\mathbf{r}_{i,u}\right)(v)/d_v}}\\
        &=\frac{\mathbf{q}_{L,u}(v)}{d_v}+\sum_{i=0}^{L}{\sum_{k=i}^{L}{\sum_{w\in\mathcal{V}}{\frac{\mathbf{r}_{i,u}(w)}{d_v}\mathbf{e}_v^T\left(\frac{1}{2}\mathbf{I}+\frac{1}{2}\mathbf{P}\right)^{k-i}\mathbf{e}_w}}}\\
        &=\frac{\mathbf{q}_{L,u}(v)}{d_v}+\sum_{i=0}^{L}{\sum_{k=i}^{L}{\sum_{w\in\mathcal{V}}{\frac{\mathbf{r}_{i,u}(w)}{d_w}\mathbf{e}_w^T\left(\frac{1}{2}\mathbf{I}+\frac{1}{2}\mathbf{P}\right)^{k-i}\mathbf{e}_v}}}\\
        &=\frac{\mathbf{q}_{L,u}(v)}{d_v}+\sum_{i=0}^{L}{\sum_{k=0}^{L-i}{\sum_{w\in\mathcal{V}}{\frac{\mathbf{r}_{i,u}(w)}{d_w}\mathbf{e}_w^T\left(\frac{1}{2}\mathbf{I}+\frac{1}{2}\mathbf{P}\right)^{k}\mathbf{e}_v}}}\\
        &=\frac{\mathbf{q}_{L,u}(v)}{d_v}+\sum_{k=0}^{L}{\sum_{w\in\mathcal{V}}{\left(\sum_{i=0}^{L-k}{\frac{\mathbf{r}_{i,u}(w)}{d_w}}\right)\mathbf{e}_w^T\left(\frac{1}{2}\mathbf{I}+\frac{1}{2}\mathbf{P}\right)^{k}\mathbf{e}_v}}.
\end{aligned}
\end{equation}
We observe that $\mathbf{e}_w^T\left(\frac{1}{2}\mathbf{I}+\frac{1}{2}\mathbf{P}\right)^{k}\mathbf{e}_v$ is the probability that a $k$-step lazy random walk start from $v$ and ends at $w$. Thus we define a sequence of random variables output by lazy random walks: $\{X_{v,k,j}\}$ with $v\in \{s,t\}, k\in [L]$ and $j\in[n_r]$. Each $X_{v,k,j}$ represents the $j^{th}$ sample of random walk, start from node $v$, and the walk length is $k$. If this $k$ step lazy random walk arrive at $w$, then $X_{v,k,j}=\sum_{i=0}^{L-k}{\frac{\mathbf{r}_{i,u}(w)}{d_w}}$. We denote $\bar{X}_{v,k}=\frac{1}{n_r}\sum_{j=1}^{n_r}{X_{v,k,j}}$ as the average of these random variables. So by Eq. (\ref{equ:rw_representation}), clearly we have:
\begin{equation}\label{equ:rw_repre_expectation}
\begin{aligned}
     \frac{\mathbf{p}_{L,u}(v)}{d_v}=\frac{\mathbf{q}_{L,u}(v)}{d_v}+\sum_{k=0}^{L}{\mathbb{E}[\bar{X}_{v,k}]}.
\end{aligned}
\end{equation}
This immedeately proves that Algorithm \ref{algo:bidir_single_pair} outputs the unbiased approximation of $\mathbf{p}_{L,u}(v)$. The next Lemma proves the concentration of this approximation.

\begin{lemma}
\label{lem:rw_concentration}
    When setting $n_r=\tilde{O}(\frac{L^2}{\epsilon^2}r_{max} d )$, we have $\left|\bar{X}_{v,k}-\mathbb{E}[\bar{X}_{v,k}]\right|\leq \frac{\epsilon}{Ld}$ with probability at least $1-n^{-1}$, for every $v\in \{s,t\}$ and $k\in [L]$.
\end{lemma}

\stitle{Proof of Lemma \ref{lem:rw_concentration}.}\hypertarget{proof_of_lemma:rw_concentration}{}
    We observe that by our implementation of \push, $\mathbf{r}_{i,u}(w)\leq \frac{r_{max}d_w}{L^2}$ for any $w\in \mathcal{V}$, $i\in[L]$. Thus, for each random variable $|X_{v,k,j}|\leq \max_{w,k} \left\{\sum_{i=0}^{L-k}{\frac{\mathbf{r}_{i,u}(w)}{d_w}}\right\}\leq \frac{r_{max}}{L}$. Next, we note that for fixed $k,v$, the random variables $\{X_{v,k,j}\}$ are i.i.d since we sample random walks independently. Thus,
    \begin{align*}
        Var[\bar{X}_{v,k}]&=\frac{1}{n_r}Var[X_{v,k,j}]\leq \frac{1}{n_r}\mathbb{E}[X_{v,k,j}^2]\\
        &\leq \frac{r_{max}}{Ln_r}\mathbb{E}[X_{v,k,j}]\leq \frac{r_{max}}{Ln_r}\frac{\mathbf{p}_{L,u}(v)}{d_v}\\
        &\leq \frac{r_{max}}{Ln_r}\frac{\Vert \mathbf{p}_{L,u}\Vert_1}{d_v}\leq \frac{r_{max}}{n_r}\frac{1}{d_v}\leq \frac{\epsilon^2}{4L^2d^2}.
    \end{align*}
    The final inequality holds when setting $n_r=\frac{4L^2}{\epsilon^2}r_{max}d$, since $d=\min\{d_s,d_t\}$ and $v\in \{s,t\}$. Next, by Chebyshev's inequality,
    \begin{align*}
        \mathbb{P}\left[\left|\bar{X}_{v,k}-\mathbb{E}[\bar{X}_{v,k}]\right|\geq \frac{\epsilon}{Ld}\right]\leq \frac{1}{2}.
    \end{align*}
    Finally, by the standard median of the means estimator, repeat the sampling process $\log n$ times and take the median as the output, we can make the final output satisfies $\left|\bar{X}_{v,k}-\mathbb{E}[\bar{X}_{v,k}]\right|\leq \frac{\epsilon}{Ld}$ with probability at least $1-n^{-\Omega(1)}$. Putting these things together, the total sampling times $n_r=\frac{4L^2}{\epsilon^2}r_{max} d \log n=\tilde{O}(\frac{L^2}{\epsilon^2}r_{max} d )$ is enough.\hfill$\blacksquare$\par

Armed with Lemma \ref{lem:rw_concentration}, we are now ready to prove Theorem \ref{thm:bidir_guarantee_specific}. 

\stitle{Proof of Theorem \ref{thm:bidir_guarantee_specific}.}\hypertarget{proof_of_theorem:bidir_guarantee_specific}{} 
From Lemma \ref{lem:rw_concentration}, for each $k\in[L]$, the inequality $\left|\bar{X}_{v,k}-\mathbb{E}[\bar{X}_{v,k}]\right|\leq \frac{\epsilon}{Ld}$ with probability at least $1-n^{-1}$. By Eq. (\ref{equ:rw_repre_expectation}), the error produced by Algorithm \ref{algo:bidir_single_pair} can be bounded by:
\begin{align*}
    \left|\frac{\mathbf{p}_{L,u}(v)}{d_v}-\frac{\hat{\mathbf{p}}_{L,u}(v)}{d_v}\right|\leq \sum_{k=0}^{L}{\left|\mathbb{E}[\bar{X}_{v,k}]-X_{v,k}\right|}\leq \sum_{k=0}^{L}{\frac{\epsilon}{Ld}}\leq \frac{\epsilon}{d}.
\end{align*}
with probability at least $1-Ln^{-1}=1-\tilde{O}(n^{-1})$. Next, by Lemma \ref{lem:guarantee_push}, the time complexity for Phase I is $O(\frac{L^3}{r_{max}})$. For the time complexity of Phase II, since we set the number of random walks $n_r=\tilde{O}(\frac{L^2}{\epsilon^2}r_{max} d )$, for each random walk the length is $L$. So the time complexity of Phase II is $O(n_r L)=\tilde{O}(\frac{L^3r_{max}d}{\epsilon^2})$. Finally, to balance the time complexity for Phase I and Phase II, we set $r_{max}=\frac{\epsilon}{\sqrt{d}}$. So the total time complexity of Algorithm \ref{algo:bidir_single_pair} can be bounded by $\tilde{O}(\frac{L^3\sqrt{d}}{\epsilon})=\tilde{O}(\frac{\sqrt{d}}{\epsilon})$. This finishes the proof of Theorem \ref{thm:bidir_guarantee_specific}. \hfill$\blacksquare$\par

Finally, we use Theorem \ref{thm:bidir_guarantee_specific} to prove Theorem \ref{thm:bidir_complexity}.

\stitle{Proof of Theorem \ref{thm:bidir_complexity}.}\hypertarget{proof_of_theorem:bidir_complexity}{}
By Equ. (\ref{equ:ER_p_expression}) and Theorem \ref{thm:bidir_guarantee_specific}, the approximation output by Algorithm \ref{algo:bidir_single_pair} satisfy: $|\hat{r}_\mathcal{G}(s,t)-r_{\mathcal{G},L}(s,t)|\leq  \frac{4\epsilon}{d}$ w.h.p. By the fact $r_\mathcal{G}(s,t)\geq \frac{1}{2}\left(\frac{1}{d_s}+\frac{1}{d_t}\right)\geq \frac{1}{2d}$, so $|\hat{r}_\mathcal{G}(s,t)-r_{\mathcal{G},L}(s,t)|\leq {8\epsilon} r_\mathcal{G}(s,t)$ w.h.p.. In addition, by Lemma \ref{lem:ER_truncate}, The $L$ step truncation ER satisfy $|r_\mathcal{G}(s,t)-r_{\mathcal{G},L}(s,t)|\leq \epsilon r_\mathcal{G}(s,t)$, so the approximation $\hat{r}_\mathcal{G}(s,t)$ safisfy:
\begin{align*}
    |\hat{r}_\mathcal{G}(s,t)-r_{\mathcal{G}}(s,t)|&\leq |\hat{r}_\mathcal{G}(s,t)-r_{\mathcal{G},L}(s,t)|+|r_\mathcal{G}(s,t)-r_{\mathcal{G},L}(s,t)|\\
    &\leq 8\epsilon r_\mathcal{G}(s,t)+\epsilon r_\mathcal{G}(s,t)=9\epsilon r_\mathcal{G}(s,t).
\end{align*}

This proves that $\hat{r}_\mathcal{G}(s,t)$ is the $9 \epsilon$-approximation of $r_\mathcal{G}(s,t)$ w.h.p. Finally, by Theorem \ref{thm:bidir_guarantee_specific}, the time complexity of Algorithm \ref{algo:bidir_single_pair} is $\tilde{O}(\frac{\sqrt{d}}{\epsilon})$.\hfill$\blacksquare$\par

\section{Theoretical analysis of Algorithm \ref{algo:bidir_sketch}}\label{sec:analysis_sketch}

In this section, we provide the theoretical analysis of Algorithm \ref{algo:bidir_sketch}, and prove Theorem \ref{thm:ER_sketch_specific} and Theorem \ref{thm:ER_sketch}. First, inspired by ~\cite{wei2024approximating} and ~\cite{li2023new}, we define the $\epsilon$-contributing set for $\mathbf{p}_{L,u}$.

\begin{definition}
    for any $u\in \mathcal{V}$, the $\epsilon$-contributing set of $\mathbf{p}_{L,u}$ is defined as $S_{\epsilon,u}=\{w\in \mathcal{V}:\mathbf{p}_{L,u}(w)> \epsilon\}$.
\end{definition}

Next, we observe that the size of $\epsilon$-contributing set is only $\tilde{O}(\frac{1}{\epsilon})$, independent of $n$. This immediately follows by Fact \ref{fact:p_L_norm_1}.

\begin{fact}\label{fact:contributing_set}
    The size of the $\epsilon$-contributing set $|S_{\epsilon,u}|\leq \frac{L}{2\epsilon}=\tilde{O}(\frac{1}{\epsilon})$.
\end{fact}

Since in Theorem \ref{thm:ER_sketch_specific} we just need to approximate the vector such that $|\hat{\mathbf{p}}_{L,u}(w)-\mathbf{p}_{L,u}(w)|\leq \epsilon$, so we only focus on the $\epsilon$-contributing set $S_{\epsilon,u}$. This is because if a node $w\notin S_{\epsilon,u}$, by definition $0\leq \mathbf{p}_{L,u}(w)\leq \epsilon$, so we set $\hat{\mathbf{p}}_{L,u}(w)=0$ is already a reasonable approximation. Next, we will prove Stage I of Algorithm \ref{algo:bidir_sketch} outputs the candidate set $S'_{\epsilon,u}$ that contains $S_{\epsilon,u}$ with high probability. 

\begin{lemma}
\label{lem:phase_1}
    $\mathbf{p}'_{L,u}$ output by Stage I satisfy the following error guarantee: 
    \begin{itemize}
    \item (i) For any $\mathbf{p}_{L,u}(w)> \epsilon$, we have 
$|\mathbf{p}'_{L,u}(w)-\mathbf{p}_{L,u}(w)|\leq \frac{1}{4}\mathbf{p}_{L,u}(w)$
with probability at least $1-n^{-\Omega(1)}$.
    
\item (ii) For any $\mathbf{p}_{L,u}(w)\leq \epsilon$, we have 
$|\mathbf{p}'_{L,u}(w)-\mathbf{p}_{L,u}(w)|\leq \frac{1}{4}\epsilon$
with probability at least $1-n^{-\Omega(1)}$.
\end{itemize}
    Furthermore, $S'_{\epsilon,u}\supset S_{\epsilon,u}$ with probability at least $1-n^{-\Omega(1)}$.
\end{lemma}

\stitle{Proof of Lemma \ref{lem:phase_1}.}\hypertarget{proof_of_lemma:phase_1}{}
    We define a sequence of random variables output by lazy random walks: $\{X^{(1)}_{u,w,k,j}\}$ with source node $u$ and target node $w$, and $ k\in [L]$ and $j\in[n_r^{(1)}]$. Each $X^{(1)}_{u,w,k,j}$ represents the $j^{th}$ sample of random walk, start from node $u$ with $k$ steps. If this $k$ step random walk reach $w$, then $X^{(1)}_{u,w,k,j}=1$, else $X^{(1)}_{u,w,k,j}=0$. We denote $\bar{X}^{(1)}_{u,w,k}=\frac{1}{n_r}\sum_{j=1}^{n_r^{(1)}}{X^{(1)}_{u,w,k,j}}$ as the average of these random variables. Thus $\mathbb{E}[\bar{X}^{(1)}_{u,w,k}]=\mathbf{e}_w^T\left(\frac{1}{2}\mathbf{I}+\frac{1}{2}\mathbf{P}\right)^k\mathbf{e}_u$ and $Var[\bar{X}^{(1)}_{u,w,k}]\leq \frac{1}{n_r^{(1)}}Var[ X_{u,w,k,j}^{(1)}]\leq \frac{1}{n_r^{(1)}}\mathbb{E}[ (X_{u,w,k,j}^{(1)})^2]= \frac{1}{n_r^{(1)}}\mathbb{E}[ X^{(1)}_{u,w,k,j}]$. We set $n_r^{(1)}= \frac{128L^2}{\epsilon }$. For $\mathbf{e}_w^T\left(\frac{1}{2}\mathbf{I}+\frac{1}{2}\mathbf{P}\right)^k\mathbf{e}_u> \frac{\epsilon}{L}$, $ k\in [L]$, by the Chebyshev's inequality,
    \begin{align*}
        \mathbb{P}\left[\left|\bar{X}^{(1)}_{u,w,k}-\mathbb{E}\left[\bar{X}^{(1)}_{u,w,k}\right]\right|\geq \frac{1}{8}\mathbb{E}\left[\bar{X}^{(1)}_{u,w,k}\right]\right]\leq \frac{1}{2}.
    \end{align*}
For $\mathbf{e}_w^T\left(\frac{1}{2}\mathbf{I}+\frac{1}{2}\mathbf{P}\right)^k\mathbf{e}_u\leq \frac{\epsilon}{L}$, $ k\in [L]$, by the Chebyshev's inequality,
\begin{align*}
        \mathbb{P}\left[\left|\bar{X}^{(1)}_{u,w,k}-\mathbb{E}\left[\bar{X}^{(1)}_{u,w,k}\right]\right|\geq \frac{\epsilon}{8L}\right]\leq \frac{1}{2}.
    \end{align*}
    
    Again, by the standard median of means estimator, repeating the sampling process $O(\log n)$ times and take the median, this probability can be improved to less than $n^{-\Omega(1)}$. Putting it together, the total sampling times $n_r^{(1)}= O(\frac{L^2\log  n}{\epsilon })$ is enough. Finally, since $\mathbf{p}_{L,u}(w)=\frac{1}{2}\sum_{k=0}^{L}{\mathbf{e}_w^T\left(\frac{1}{2}\mathbf{I}+\frac{1}{2}\mathbf{P}\right)^k\mathbf{e}_u}= \frac{1}{2}\sum_{k=0}^{L}{\mathbb{E}\left[\bar{X}^{(1)}_{u,w,k}\right]}$, we define the approximation $\mathbf{p}'_{L,u}(w)= \frac{1}{2}\sum_{k=0}^{L}{\bar{X}^{(1)}_{u,w,k}}$ (This is equivalent to the process of Line 4 in Algorithm \ref{algo:bidir_sketch}). Thus, the approximation $\hat{\mathbf{p}}_{L,u}$ satisfy the following two conditions: 
\begin{itemize}
\item (i) For any $\mathbf{p}_{L,u}(w)> \epsilon$, we have 
\begin{align*}
|\mathbf{p}'_{L,u}(w)-\mathbf{p}_{L,u}(w)|&\leq \sum_{k=0}^{L}{\left|\bar{X}^{(1)}_{u,w,k}-\mathbb{E}\left[\bar{X}^{(1)}_{u,w,k}\right]\right|}
\leq \frac{1}{8}\mathbf{p}_{L,u}(w)+\frac{\epsilon}{8}\leq \frac{1}{4}\mathbf{p}_{L,u}(w)
\end{align*}
with probability at least $1-L n^{-\Omega(1)}=1-n^{-\Omega(1)}$.
    
\item (ii) For any $\mathbf{p}_{L,u}(w)\leq \epsilon$, we have 
\begin{align*}
|\mathbf{p}'_{L,u}(w)-\mathbf{p}_{L,u}(w)|&\leq \sum_{k=0}^{L}{\left|\bar{X}^{(1)}_{u,w,k}-\mathbb{E}\left[\bar{X}^{(1)}_{u,w,k}\right]\right|}
\leq \frac{1}{8}\mathbf{p}_{L,u}(w)+\frac{\epsilon}{8}\leq \frac{1}{4}\epsilon
\end{align*}
with probability at least $1-L n^{-\Omega(1)}=1-n^{-\Omega(1)}$.
\end{itemize}
As a result, $S'_{\epsilon,u}=\{w\in \mathcal{V}:\mathbf{p}'_{L,u}(w)> \epsilon/2\}$ contains $S_{\epsilon,u}$ with probability at least $1-n^{-\Omega(1)}$.\hfill$\blacksquare$\par

Next, we perform \push operations for each node $v\in S'_{\epsilon,u}$. For each $v$, we set the threshold for \push to be $\frac{r_{max}}{2 \mathbf{p}'_{L,u}(v)}$. By Lemma \ref{lem:invariant} and Eq. (\ref{equ:rw_representation}), we have the following equation:
\begin{equation}
\begin{aligned}
        \frac{\mathbf{p}_{L,v}(u)}{d_u}=\frac{\mathbf{q}_{L,v}(u)}{d_u}+\sum_{k=0}^{L}{\sum_{w\in\mathcal{V}}{\left(\sum_{i=0}^{L-k}{\frac{\mathbf{r}_{i,v}(w)}{d_w}}\right)\mathbf{e}_w^T\left(\frac{1}{2}\mathbf{I}+\frac{1}{2}\mathbf{P}\right)^{k}\mathbf{e}_u}}.
\end{aligned}
\end{equation}
    By the fact $\frac{\mathbf{p}_{L,v}(u)}{d_u}=\frac{\mathbf{p}_{L,u}(v)}{d_v}$, this is equivalent to
    \begin{equation}
    \begin{aligned}
        \mathbf{p}_{L,u}(v)=\frac{\mathbf{q}_{L,v}(u)d_v}{d_u}+d_v\sum_{k=0}^{L}{\sum_{w\in\mathcal{V}}{\left(\sum_{i=0}^{L-k}{\frac{\mathbf{r}_{i,v}(w)}{d_w}}\right)\mathbf{e}_w^T\left(\frac{1}{2}\mathbf{I}+\frac{1}{2}\mathbf{P}\right)^{k}\mathbf{e}_u}}.
    \end{aligned}
    \end{equation}
    This enables us to approximate $\mathbf{p}_{L,u}(v)$ for any $v\in S'_{\epsilon,u}$ by random walks from just a source node $u$, see Fig. \ref{fig:our_algo} (b) as an illustration. Similar to the proof for Lemma \ref{lem:rw_concentration} and Lemma \ref{lem:phase_1}, we define a sequence of random variables: $\{X^{(3)}_{u,k,j}\}$. Each $X^{(3)}_{u,k,j}$ represents the $j^{th}$ sample of random walk, start from node $u$ with $k$ steps. If this $k$ step random walk reach $w$, then we set $X^{(3)}_{u,k,j}=\sum_{i=0}^{L-k}{\frac{\mathbf{r}_{i,v}(w)}{d_w}}$. We denote $\bar{X}^{(3)}_{u,k}=\frac{1}{n_r^{(3)}}\sum_{j=1}^{n_r^{(3)}}{X^{(3)}_{u,k,j}}$ as the average of these random variables. For each $v\in S'_{\epsilon,u}$, clearly we have
\begin{equation}
    \begin{aligned}
        \mathbf{p}_{L,u}(v)=\frac{\mathbf{q}_{L,v}(u)d_v}{d_u}+d_v\sum_{k=0}^{L}{\mathbb{E}[\bar{X}^{(3)}_{u,k}]}.
    \end{aligned}
    \end{equation}

We define the approximation as
\begin{equation}
    \begin{aligned}
\hat{\mathbf{p}}_{L,u}(v)= \frac{\mathbf{q}_{L,v}(u)d_v}{d_u}+d_v\sum_{k=0}^{L}{\bar{X}^{(3)}_{u,k}}.
\end{aligned}
    \end{equation}

This is equivalent to the expression of Line 12 in Algorithm \ref{algo:bidir_sketch}. Next, we prove the concentration result of $\hat{\mathbf{p}}_{L,u}$.

\begin{lemma}
\label{lem:rw_sketch_concentration}
    When setting $n_r^{(3)}=\tilde{O}(\frac{L}{\epsilon^2}r_{max}d_u)$, the approximation $\hat{\mathbf{p}}_{L,u}$ satisfy: $|\hat{\mathbf{p}}_{L,u}(v)-\mathbf{p}_{L,u}(v)|\leq \epsilon$ with probability at least $1-n^{-\Omega(1)}$ for any $v\in S'_{\epsilon,u}$ with $d_v\leq d_u$.
\end{lemma}

\stitle{Proof of Lemma \ref{lem:rw_sketch_concentration}.} \hypertarget{proof_of_lemma:rw_sketch_concentration}{}
we set the threshold for \push $\frac{r_{max}}{2 \mathbf{p}'_{L,u}(v)}$, thus $r_{i,v}(w)\leq \frac{r_{max}d_w}{2L^2\mathbf{p}'_{L,u}(v)}$ for any $w\in \mathcal{V}$ and $i\leq L$. For each random variable $X^{(3)}_{u,k,j}$, we have $|X^{(3)}_{u,k,j}|\leq \max_{w,k} \left\{\sum_{i=0}^{L-k}{\frac{\mathbf{r}_{i,v}(w)}{d_w}}\right\}\leq \frac{r_{max}}{2L\mathbf{p}'_{L,u}(v)}$. Therefore, 
    \begin{align*}
        Var[\bar{X}^{(3)}_{u,k}]&=\frac{1}{n_r}Var[X^{(3)}_{u,k,j}]\leq \frac{1}{n_r}\mathbb{E}[(X^{(3)}_{u,k,j})^2]\\
        &\leq \frac{r_{max}}{2Ln_r\mathbf{p}'_{L,u}(v)}\mathbb{E}[X^{(3)}_{u,k,j}]\leq \frac{r_{max}}{2Ln_r\mathbf{p}'_{L,u}(v)}\frac{\mathbf{p}_{L,u}(v)}{d_v}\\
        &\leq \frac{r_{max}}{Ln_r}\frac{1}{d_v}\leq \frac{\epsilon^2}{4L^2d_ud_v},
    \end{align*}
where the final inequality holds when setting $n_r^{(3)}=\frac{4L}{\epsilon^2}r_{max}d_u$. By $d_v\leq d_u$, we obtain $Var[\bar{X}^{(3)}_{u,k}]\leq \frac{\epsilon^2}{4L^2d_v^2}$, which is equivalent to $Var[d_v\bar{X}^{(3)}_{u,k}]\leq \frac{\epsilon^2}{4L^2}$. By Chebyshev's inequality,
\begin{equation}\label{equ:sketch_var}
\begin{aligned}
    \mathbb{P}[|d_v\bar{X}^{(3)}_{u,k}-\mathbb{E}[d_v\bar{X}^{(3)}_{u,k}]|\geq \frac{\epsilon}{L}]\leq \frac{1}{2}.
\end{aligned}
\end{equation}
Again, by the standard median of means estimator, we repeat the sampling process $O(\log n)$ times to improve this probability to less than $n^{-\Omega(1)}$. The total sampling times is $n_r^{(3)}=O(\frac{L}{\epsilon^2}r_{max}d_u\log  n)=\tilde{O}(\frac{L}{\epsilon^2}r_{max}d_u)$. Since the approximation is defined as $\hat{\mathbf{p}}_{L,u}(v)= \frac{\mathbf{q}_{L,v}(u)d_v}{d_u}+d_v\sum_{k=0}^{L}{\bar{X}^{(3)}_{u,k}}$, by Equ. (\ref{equ:sketch_var}), 
\begin{align*}
|\hat{\mathbf{p}}_{L,u}(v)-\mathbf{p}_{L,u}(v)|&\leq \sum_{k=0}^{L}{|d_v\bar{X}^{(3)}_{u,k}-\mathbb{E}[d_v\bar{X}^{(3)}_{u,k}]|}\leq \sum_{k=0}^{L}{\frac{\epsilon}{L}}\leq \epsilon
\end{align*}
with probability at least $1-L n^{-\Omega(1)}=1-n^{-\Omega(1)}$. This finishes the proof.\hfill$\blacksquare$\par

Finally, we use Lemma \ref{lem:rw_sketch_concentration} to prove Theorem \ref{thm:ER_sketch_specific}.

\stitle{Proof of Theorem \ref{thm:ER_sketch_specific}.}\hypertarget{proof_of_theorem:ER_sketch_specific}{} 
First, the error bound of Algorithm \ref{algo:bidir_sketch} directly follows by Lemma \ref{lem:rw_sketch_concentration}. For the time complexity of Algorithm \ref{algo:bidir_sketch}, for each $u\in \mathcal{V}$, in Phase I we perform $n_r^{(1)}$ random walks, each length $L$, so the time complexity is $O(n_r^{(1)}L)=O(\frac{4L^3\log n}{\epsilon})=\tilde{O}(\frac{L^3}{\epsilon})$. In Phase II we perform \push from the candidate set $S'_{\epsilon,u}$ with threshold $\frac{r_{max}}{2\mathbf{p}'_{L,u}}$, the time complexity is bounded by $O(\sum_{v\in S'_{\epsilon,u}}{\frac{L^3\mathbf{p}'_{L,u}(v)}{r_{max}}})=O(\sum_{v\in S'_{\epsilon,u}}{\frac{L^3\mathbf{p}_{L,u}(v)}{r_{max}}})$. By the fact that $\Vert \mathbf{p}_{L,u}\Vert_1\leq \frac{L}{2}$, this is further bounded by $O(\frac{L^4}{r_{max}})$. In phase III we perform $n_r^{(3)}$ random walks from $u$, each length $L$, so the time complexity is $O(n_r^{(3)}L)=O(\frac{L^2}{\epsilon^2}r_{max}d_u\log  n)=\tilde{O}(\frac{L^2}{\epsilon^2}r_{max}d_u)$. We sum over all $u\in \mathcal{V}$, and setting $r_{max}=L\epsilon\sqrt{\frac{n}{m}}$, the total time complexity is:
\begin{align*}
    \tilde{O}\left(n\frac{L^3}{\epsilon}+n\frac{L^4}{r_{max}}+m\frac{L^2}{\epsilon^2}r_{max}\right)=\tilde{O}\left(n\frac{L^3}{\epsilon}+\sqrt{nm}\frac{L^3}{\epsilon}\right)=\tilde{O}\left(\frac{\sqrt{nm}}{\epsilon}\right).
\end{align*}

For the space complexity, by Fact \ref{fact:contributing_set} and Lemma \ref{lem:phase_1}, the size of the candidate set $S'_{\epsilon,u}$ can be bounded by $|S'_{\epsilon,u}|\leq \tilde{O}(\frac{1}{\epsilon})$, so the total space complexity is $\tilde{O}(\sum_{u\in \mathcal{V}}{|S'_{\epsilon,u}|})=\tilde{O}(\frac{n}{\epsilon})$. This proves Theorem \ref{thm:ER_sketch_specific}.\hfill$\blacksquare$\par

Arming with the query algorithm (Algorithm \ref{algo:sketch_query}) and Theorem \ref{thm:ER_sketch_specific}, we are ready to prove Theorem \ref{thm:ER_sketch}.

\stitle{Proof of Theorem \ref{thm:ER_sketch}.}\hypertarget{proof_of_theorem:ER_sketch}{}
Given any vertex pair $(s,t)$, without loss of generality we assume $d_s\leq d_t$. By Equ. (\ref{equ:ER_p_expression}) and the fact $\frac{\mathbf{p}_{L,s}(t)}{d_t}=\frac{\mathbf{p}_{L,t}(s)}{d_s}$, we can express the ER value by
\begin{align*}
    r_{\mathcal{G},L}(s,t)=\frac{\mathbf{p}_{L,s}(s)}{d_s}-2\frac{\mathbf{p}_{L,t}(s)}{d_s}+\frac{\mathbf{p}_{L,t}(t)}{d_t}.
\end{align*}

By Theorem \ref{thm:ER_sketch_specific}, Algorithm \ref{algo:bidir_sketch} outputs the approximation such that $|\mathbf{p}_{L,s}(s)-\hat{\mathbf{p}}_{L,s}(s)|\leq \epsilon$, $|\mathbf{p}_{L,t}(s)-\hat{\mathbf{p}}_{L,t}(s)|\leq \epsilon$ and $|\mathbf{p}_{L,t}(t) -\hat{\mathbf{p}}_{L,t}(t)|\leq \epsilon$. Thus the approximation $\hat{r}_\mathcal{G}(s,t)$ output by Algorithm \ref{algo:sketch_query} satisfy: $|\hat{r}_\mathcal{G}(s,t)-r_{\mathcal{G},L}(s,t)|\leq \frac{4\epsilon}{d_s}\leq 8\epsilon r_\mathcal{G}(s,t)$. By Lemma \ref{lem:ER_truncate}, $|r_\mathcal{G}(s,t)-r_{\mathcal{G},L}(s,t)|\leq \epsilon r_\mathcal{G}(s,t)$, thus $|\hat{r}_\mathcal{G}(s,t)-r_\mathcal{G}(s,t)|\leq 9\epsilon r_\mathcal{G}(s,t)$. In addition, since the query time of Algorithm \ref{algo:sketch_query} is $O(1)$ for a given vertex pair $(s,t)$, so our algorithm is an $(\tilde{O}(\frac{\sqrt{nm}}{\epsilon}), O(1),\tilde{O}(\frac{n}{\epsilon}))$-ER sketch algorithm. Finally, we combine our results with a previous resistance sparsifier result from ~\cite{chu2020graph} .

\begin{theorem}{~\cite{chu2020graph}}
    Given graph $\mathcal{G}$ with $m$ edges and $n$ vertices, there exist an algorithm that runs in $\tilde{O}(m)$ time and produce a graph $\mathcal{H}$ with $\tilde{O}(\frac{n}{\epsilon})$ edges, such that for $\forall s,t\in \mathcal{V}$, $r_{\mathcal{G}}(s,t)\approx_\epsilon r_\mathcal{H}(s,t)$.
\end{theorem}

Therefore, our algorithm can be also processed in $\tilde{O}(m+n/\epsilon^{1.5})$ time. Putting it together, our algorithm is an $(\tilde{O}(\min \{m+n/\epsilon^{1.5}, \sqrt{nm}/\epsilon\}),O(1),\tilde{O}(n/\epsilon))$-ER sketch algorithm.\hfill$\blacksquare$\par

\stitle{Remark} We remark that as a corollary, our $\tilde{O}(\frac{\sqrt{nm}}{\epsilon})$ time algorithm provides the first sublinear algorithm for the multiple pair ER computation problem. Besides, we remark that if we simply use \push to approximate $\mathbf{p}_{L,u}$ with threshold $\frac{\epsilon}{d_u}$ for each $u\in \mathcal{V}$, by Lemma \ref{lem:guarantee_push}, this also satisfies the error guarantee of Theorem \ref{thm:ER_sketch_specific}. 
In addition, by Lemma \ref{lem:guarantee_push}, the total time complexity of this procedure is bounded by $O(\sum_{u\in \mathcal{V}}{\frac{L^3d_u}{\epsilon}})=\tilde{O}(\frac{m}{\epsilon})$, which is asymptotically the same processing time as the Algorithm in ~\cite{dwaraknath2024towards} for expander graphs. This is an interesting corollary because \push is a deterministic algorithm that achieves the same complexity as the previous state-of-the-art randomized algorithms.


\section{Other Omitting Proofs}\label{sec:other_proofs}

\stitle{Proof of Lemma \ref{lemma:ER_taylor_expansion}.}\hypertarget{proof_of_lemma:ER_taylor_expansion}{}
    We use the similar argument as Lemma 4.3 in ~\cite{peng2021local}. First, we note that $\mathbf{L}=\mathbf{D}-\mathbf{A}=\mathbf{D}^{1/2}(\mathbf{I}-\mathcal{A})\mathbf{D}^{1/2}$, where $\mathcal{A}$ is the normalized adjacency matrix. Therefore, we have:
 \begin{align*}
    \mathbf{L}^\dagger&=\mathbf{D}^{-1/2}(\mathbf{I}-\mathcal{A})^\dagger\mathbf{D}^{-1/2}
    =\mathbf{D}^{-1/2}\sum_{j=2}^{n}{\frac{1}{1-\lambda_j(\mathcal{A})}\mathbf{u}_j\mathbf{u}_j^T}\mathbf{D}^{-1/2}\\
    &=\frac{1}{2}\mathbf{D}^{-1/2}\sum_{j=2}^{n}{\sum_{l=0}^{\infty}{\left(\frac{1}{2}+\frac{1}{2}\lambda_j(\mathcal{A})\right)^l}\mathbf{u}_j\mathbf{u}_j^T}\mathbf{D}^{-1/2}\\
    &=\frac{1}{2}\mathbf{D}^{-1/2}\sum_{l=0}^{\infty}{\sum_{j=2}^{n}{\left(\frac{1}{2}+\frac{1}{2}\lambda_j(\mathcal{A})\right)^l}\mathbf{u}_j\mathbf{u}_j^T}\mathbf{D}^{-1/2}\\
    &=\frac{1}{2}\mathbf{D}^{-1/2}\sum_{l=0}^{\infty}{\left(\left(\frac{1}{2}\mathbf{I}+\frac{1}{2}\mathcal{A}\right)^l-\mathbf{u}_1\mathbf{u}_1^T\right)}\mathbf{D}^{-1/2}.\\
 \end{align*}
 Where $\mathbf{u}_j$ denotes the eigenvector of $\mathcal{A}$ corresponding to the eigenvalue $\lambda_j(\mathcal{A})$. Next, we notice the fact that $\mathbf{u}_1=\mathbf{D}^{1/2}\mathbf{1}\perp \mathbf{D}^{-1/2}(\mathbf{e}_s-\mathbf{e}_t)$ for any $s,t\in \mathcal{V}$, where $\mathbf{1}$ is the all-one vector. Therefore,
 \begin{align*}
     r_\mathcal{G}(s,t)&=(\mathbf{e}_s-\mathbf{e}_t)^T\mathbf{L}^{\dagger}(\mathbf{e}_s-\mathbf{e}_t)\\
     &=\frac{1}{2}(\mathbf{e}_s-\mathbf{e}_t)^T\mathbf{D}^{-1/2}\sum_{l=0}^{\infty}{\left(\left(\frac{1}{2}\mathbf{I}+\frac{1}{2}\mathcal{A}\right)^l-\mathbf{u}_1\mathbf{u}_1^T\right)}\mathbf{D}^{-1/2}(\mathbf{e}_s-\mathbf{e}_t)\\
     &=\frac{1}{2}(\mathbf{e}_s-\mathbf{e}_t)^T\mathbf{D}^{-1/2}\sum_{l=0}^{+\infty}{\left(\frac{1}{2}\mathbf{I}+\frac{1}{2}\mathcal{A}\right)^l}\mathbf{D}^{-1/2}(\mathbf{e}_s-\mathbf{e}_t)\\
     &=\frac{1}{2}(\mathbf{e}_s-\mathbf{e}_t)^T\mathbf{D}^{-1}\sum_{l=0}^{+\infty}{\left(\frac{1}{2}\mathbf{I}+\frac{1}{2}\mathbf{P}\right)^l}(\mathbf{e}_s-\mathbf{e}_t).
 \end{align*}
 This finishes the proof.\hfill$\blacksquare$\par

\stitle{Proof of Lemma \ref{lem:ER_truncate}.}\hypertarget{proof_of_lemma:ER_truncate}{}
    By the definition of $r_{\mathcal{G},L}(s,t)$, we have that
    \begin{align*}
        &|r_\mathcal{G}(s,t)-r_{\mathcal{G},L}(s,t)|\\
        &= \left|\frac{1}{2}(\mathbf{e}_s-\mathbf{e}_t)^T\mathbf{D}^{-1}\sum_{l=L+1}^{\infty}{\left(\frac{1}{2}\mathbf{I}+\frac{1}{2}\mathbf{P}\right)^l}(\mathbf{e}_s-\mathbf{e}_t)\right|\\
        &=\left|\frac{1}{2}(\mathbf{e}_s-\mathbf{e}_t)^T\mathbf{D}^{-1/2}\sum_{l=L+1}^{\infty}{\left(\frac{1}{2}\mathbf{I}+\frac{1}{2}\mathcal{A}\right)^l}\mathbf{D}^{-1/2}(\mathbf{e}_s-\mathbf{e}_t)\right|\\
        &=\left|\frac{1}{2}(\mathbf{e}_s-\mathbf{e}_t)^T\mathbf{D}^{-1/2}\sum_{l=L+1}^{\infty}{\sum_{j=2}^{n}{\left(\frac{1}{2}+\frac{1}{2}\lambda_j(\mathcal{A})\right)^l}\mathbf{u}_j\mathbf{u}_j^T}\mathbf{D}^{-1/2}(\mathbf{e}_s-\mathbf{e}_t)\right|.
    \end{align*}

    Next, we denote $\alpha_j=\mathbf{u}_j^T\mathbf{D}^{-1/2}(\mathbf{e}_s-\mathbf{e}_t)$, and $\alpha=(\alpha_1,...,\alpha_n)^T=(\mathbf{u}_1,...,\mathbf{u}_n)^T\mathbf{D}^{-1/2}(\mathbf{e}_s-\mathbf{e}_t)$. Therefore, we can bound
    \begin{align*}
    \Vert\alpha\Vert_2&=\Vert(\mathbf{u}_1,...,\mathbf{u}_n)\mathbf{D}^{-1/2}(\mathbf{e}_s-\mathbf{e}_t)\Vert_2
    \leq \Vert\mathbf{D}^{-1/2}(\mathbf{e}_s-\mathbf{e}_t)\Vert_2\leq \sqrt{2}.
    \end{align*}
    And thus $\sum_{j=1}^{n}{\alpha_j^2}\leq 2$. Therefore,
    \begin{align*}
        |r_\mathcal{G}(s,t)-r_{\mathcal{G},L}(s,t)|&=\left|\frac{1}{2}\sum_{l=L+1}^{\infty}{\sum_{j=2}^{n}{\left(\frac{1}{2}+\frac{1}{2}\lambda_j(\mathcal{A})\right)^l}}\alpha_j^2\right|
        \leq \sum_{l=L+1}^{\infty}{\left(\frac{1}{2}+\frac{1}{2}\lambda_2(\mathcal{A})\right)^l},\\
    \end{align*}
    where $\lambda_2(\mathcal{A})$ is the second largest eigenvalue of $\mathcal{A}$. By the property of the eigenvalues of normalized Laplacian matrix, $\lambda_2(\mathcal{L})=1-\lambda_2(\mathcal{A})$. Therefore,
    \begin{align*}
        |r_\mathcal{G}(s,t)-r_{\mathcal{G},L}(s,t)|
        &\leq \sum_{l=L+1}^{\infty}{\left(1-\frac{1}{2}\lambda_2(\mathcal{L})\right)^l}
        \leq \left(1-\frac{1}{2}\lambda_2(\mathcal{L})\right)^L\frac{2}{\lambda_2(\mathcal{L})} \\
        &\leq \left(1-\frac{1}{2}\lambda_2(\mathcal{L})\right)^L\kappa(\mathcal{L}) \leq \frac{\epsilon}{n}.
    \end{align*}
    The final inequality holds when setting $L=\log \frac{\epsilon}{n\kappa(\mathcal{L})}/ \log \left(1-\frac{1}{2}\lambda_2(\mathcal{L})\right)$. We notice the fact $\log \left(1-\frac{1}{2}\lambda_2(\mathcal{L})\right)\leq  -\frac{1}{2}\lambda_2(\mathcal{L})$ and $\kappa(\mathcal{L})\leq \frac{2}{\lambda_2(\mathcal{L})}$, so we set $L\geq 2\kappa(\mathcal{L}) \log \frac{n}{\epsilon}$ is enough. Finally, by the fact $r_\mathcal{G}(s,t)\geq \frac{1}{d_s}+\frac{1}{d_t}\geq \frac{1}{n}$, we have $ |r_\mathcal{G}(s,t)-r_{\mathcal{G},L}(s,t)|\leq \frac{\epsilon}{n}\leq\epsilon r_\mathcal{G}(s,t)$, thus $r_{\mathcal{G},L}(s,t)$ is the $\epsilon$-approximation of $r_\mathcal{G}(s,t)$.\hfill$\blacksquare$\par

   \stitle{Proof of Fact \ref{fact:p_L_norm_1}.}\hypertarget{proof_of_fact:p_L_norm_1}{}
   By Definition \ref{def:p_L}, 
    \begin{align*}
        \Vert \mathbf{p}_{L,u}\Vert_1&=\left\Vert \frac{1}{2}\sum_{l=0}^{L}{\left(\frac{1}{2}\mathbf{I}+\frac{1}{2}\mathbf{P}\right)^l}\mathbf{e}_u\right\Vert_1
        \leq  \frac{1}{2}\sum_{l=0}^{L}{\left\Vert\left(\frac{1}{2}\mathbf{I}+\frac{1}{2}\mathbf{P}\right)^l\mathbf{e}_u\right\Vert_1}\leq  \frac{L}{2}.
    \end{align*}
    The final inequality holds because $\left(\frac{1}{2}\mathbf{I}+\frac{1}{2}\mathbf{P}\right)^l$ is a column stochastic matrix.\hfill$\blacksquare$\par

\stitle{Proof of Fact \ref{fact:p_L_transform}.}
    By Definition \ref{def:p_L}, 
    \begin{align*}
        \frac{\mathbf{p}_{L,u}(v)}{d_v}&=\frac{1}{2}\sum_{l=0}^{L}{\frac{1}{d_v}\left(\left(\frac{1}{2}\mathbf{I}+\frac{1}{2}\mathbf{P}\right)^l\mathbf{e}_u\right)(v)}
        =\frac{1}{2}\sum_{l=0}^{L}{\frac{1}{d_v}\mathbf{e}_v^T\left(\frac{1}{2}\mathbf{I}+\frac{1}{2}\mathbf{P}\right)^l\mathbf{e}_u}\\
        &=\frac{1}{2}\sum_{l=0}^{L}{\frac{1}{d_u}\mathbf{e}_u^T\left(\frac{1}{2}\mathbf{I}+\frac{1}{2}\mathbf{P}\right)^l\mathbf{e}_v}=\frac{\mathbf{p}_{L,v}(u)}{d_u}.
    \end{align*}
\hfill$\blacksquare$\par

\stitle{Proof of Fact \ref{fact:ER_lower_bound}.}\hypertarget{proof_of_fact:ER_lower_bound}{}
    We notice the fact that $\mathbf{L}=\mathbf{D}-\mathbf{A}\preceq 2\mathbf{D}$. Thus, for any $\mathbf{x}\perp \mathbf{1}$, $\mathbf{x}^T\mathbf{L}^\dagger \mathbf{x}\geq \mathbf{x}^T\frac{1}{2}\mathbf{D}^{-1} \mathbf{x}$. Therefore, 
    \begin{align*}
        r_\mathcal{G}(s,t)&=(\mathbf{e}_s-\mathbf{e}_t)^T\mathbf{L}^{\dagger}(\mathbf{e}_s-\mathbf{e}_t)\\
        &\geq (\mathbf{e}_s-\mathbf{e}_t)^T\frac{1}{2}\mathbf{D}^{-1}(\mathbf{e}_s-\mathbf{e}_t)=\frac{1}{2}\left(\frac{1}{d_s}+\frac{1}{d_t}\right).
    \end{align*}
\hfill$\blacksquare$\par






\stitle{Proof of Lemma \ref{lem:parallel_resistance}.}\hypertarget{proof_of_lemma:parallel_resistance}{}
To prove Lemma \ref{lem:parallel_resistance}, we need the following two Propositions. The first Proposition is the submatrix representation of ER value, and the second Proposition is the classic block matrix inverse formula.

\begin{proposition}{(Corollary 1 in ~\cite{liao2023resistance})}\label{prop:landmark}
    $r_\mathcal{G}(s,t)=(\mathbf{e}_s-\mathbf{e}_t)^T\mathbf{L}^\dagger(\mathbf{e}_s-\mathbf{e}_t)=(\mathbf{L}_s^{-1})_{tt}=( \mathbf{L}_t^{-1})_{ss}$ for any graph $\mathcal{G}$. Where $\mathbf{L}$ denotes the Laplacian matrix of $\mathcal{G}$, and $\mathbf{L}_s$ (resp., $\mathbf{L}_t$) denotes the submatrix of $\mathbf{L}$, deleting the row and column indexed by node $s$ (resp., $t$).
\end{proposition}

\begin{proposition}\label{prop:block_inverse}
      Let $S=D-CA^{-1}B$, we have:
      \begin{align*}
      \begin{bmatrix}
       \begin{array}{cc}
       A & B\\
       C & D
        \end{array}
        \end{bmatrix}^{-1}=
        \begin{bmatrix}
       \begin{array}{cc}
       A^{-1}+A^{-1}BS^{-1}CA^{-1} & -A^{-1}BS^{-1}\\
       -S^{-1}CA^{-1} & S^{-1}
        \end{array}
        \end{bmatrix}
        \end{align*}
\end{proposition}

Now we use Proposition \ref{prop:landmark} and Proposition \ref{prop:block_inverse} to prove Lemma \ref{lem:parallel_resistance}. We observe that the resistance $r_{S_2}(s,t)$ is simply the the parallel connection of $(d_s-x)$ resistors with resistance value $2$. Thus, by the parallel resistance formula, $r_{S_2}(s,t)=\frac{2}{d_s-x}$. For $r_{S_1}(s,t)$, we define $\mathbf{L}_{S_1\cup\{s,t\}}\in \mathbb{R}^{(n_1+2)\times(n_1+2)}$ as the Laplacian matrix constraint on the edges between $S_1\cup \{s,t\}$, $\mathbf{L}_{S_1}\in \mathbb{R}^{n_1\times n_1}$ as the Laplacian matrix constraint on the edges between $S_1$. Let $\mathbf{e}_{\mathcal{N}_1}$ denotes the indicating vector that only takes value $1$ in $\mathcal{N}_1$ and $0$ elsewhere, $\mathbf{I}_{\mathcal{N}_1}$ denotes the diagonal matrix that only takes value $1$ at diagonal entries indexed by $\mathcal{N}_1$ and $0$ elsewhere, where $ \mathcal{N}_1=\{v_1,...,v_x\}$ is the set of the neighbors of $s$ in $S_1$. By the definition of ER and the submatrix representation (Proposition \ref{prop:landmark}), we have:
\begin{align*}
    r_{S_1}(s,t)&=(\mathbf{e}_s-\mathbf{e}_t)^T\mathbf{L}_{S_1\cup\{s,t\}}^\dagger(\mathbf{e}_s-\mathbf{e}_t)
    =\begin{bmatrix}
       \begin{array}{cc}
       x & -\mathbf{e}_{\mathcal{N}_1}^T\\
       -\mathbf{e}_{\mathcal{N}_1} & \mathbf{L}_{S_1}+\mathbf{I}+\mathbf{I}_{\mathcal{N}_1}
        \end{array}
        \end{bmatrix}_{s,s}^{-1}
\end{align*}

Next, we use the block matrix inverse formula (Proposition \ref{prop:block_inverse}), we have:
\begin{equation}\label{equ:compute_lower_bound}
\begin{aligned}
    & \begin{bmatrix}
       \begin{array}{cc}
       x & -\mathbf{e}_{\mathcal{N}_1}^T\\
       -\mathbf{e}_{\mathcal{N}_1} & \mathbf{L}_{S_1}+\mathbf{I}+\mathbf{I}_{\mathcal{N}_1}
        \end{array}
        \end{bmatrix}_{s,s}^{-1}
        =\frac{1}{x}+\frac{1}{x^2}\mathbf{e}_{\mathcal{N}_1}^T\left[\mathbf{L}_{S_1}+\mathbf{I}+\mathbf{I}_{\mathcal{N}_1}-\frac{1}{x}\mathbf{e}_{\mathcal{N}_1}\mathbf{e}_{\mathcal{N}_1}^T\right]^{-1}\mathbf{e}_{\mathcal{N}_1}
\end{aligned}
\end{equation}

Now we prove the second term of the right hand side of Eq. (\ref{equ:compute_lower_bound}) is at most $O(\frac{\log n_1}{dx})$. To this end, we notice that $|\mathcal{N}_1|=x$ and $\mathbf{I}_{\mathcal{N}_1}-\frac{1}{x}\mathbf{e}_{\mathcal{N}_1}\mathbf{e}_{\mathcal{N}_1}^T\succeq 0$, thus
\begin{align*}
    \mathbf{L}_{S_1}+\mathbf{I}+\mathbf{I}_{\mathcal{N}_1}-\frac{1}{x}\mathbf{e}_{\mathcal{N}_1}\mathbf{e}_{\mathcal{N}_1}^T\succeq \mathbf{L}_{S_1}+\mathbf{I}.
\end{align*}

Next, we define $\mathbf{A}_{S_1},\mathbf{P}_{S_1}\in \mathbb{R}^{n_1\times n_1}$ to be the adjacency matrix and probability transition matrix constraint on the edges between $S_1$, respectively. By our construction, we have
\begin{align*}
    \mathbf{L}_{S_1}+\mathbf{I}&=(d+1)\mathbf{I}-\mathbf{A}_{S_1}=(d+1)\mathbf{I}-d\mathbf{P}_{S_1}\\
    &=(d+1)\left[\mathbf{I}-\left(1-\frac{1}{d+1}\right)\mathbf{P}_{S_1}\right].
\end{align*}
Therefore, we can bound:
\begin{equation}\label{equ:compute_lower_bound_2}
\begin{aligned}
    0 \leq \mathbf{e}_{\mathcal{N}_1}^T\left[\mathbf{L}_{S_1}+\mathbf{I}+\mathbf{I}_{\mathcal{N}_1}-\frac{1}{x}\mathbf{e}_{\mathcal{N}_1}\mathbf{e}_{\mathcal{N}_1}^T\right]^{-1}\mathbf{e}_{\mathcal{N}_1} 
    &\leq 
    \frac{1}{d+1}\mathbf{e}_{\mathcal{N}_1}^T\left[\mathbf{I}-\left(1-\frac{1}{d+1}\right)\mathbf{P}_{S_1}\right]^{-1}\mathbf{e}_{\mathcal{N}_1}\\
    &=\frac{1}{d+1}\sum_{k=0}^{\infty}{\left(1-\frac{1}{d+1}\right)^k\mathbf{e}_{\mathcal{N}_1}^T\mathbf{P}_{S_1}^{k}\mathbf{e}_{\mathcal{N}_1}}.
\end{aligned}
\end{equation}

Next we use the classical random walk mixing result. For expander graph $S_1$, if we start a random walk with initial probability $\mathbf{\pi}_{\mathcal{N}_1}=\frac{1}{|\mathcal{N}_1|} \mathbf{e}_{\mathcal{N}_1}$, after $K=O(\log n_1)$ steps, $\left\Vert \mathbf{P}_{S_1}^K\mathbf{\pi}_{\mathcal{N}_1}-\frac{\mathbf{1}}{n}\right\Vert_2\leq \frac{1}{n^2}$. Combing this result with Eq. (\ref{equ:compute_lower_bound_2}), we have:
\begin{align*}
&\mathbf{e}_{\mathcal{N}_1}^T\left[\mathbf{L}_{S_1}+\mathbf{I}+\mathbf{I}_{\mathcal{N}_1}-\frac{1}{x}\mathbf{e}_{\mathcal{N}_1}\mathbf{e}_{\mathcal{N}_1}^T\right]^{-1}\mathbf{e}_{\mathcal{N}_1}\\
&\leq 
\frac{1}{d+1}\left[\sum_{k=0}^{K}{\left(1-\frac{1}{d+1}\right)^k\mathbf{e}_{\mathcal{N}_1}^T\mathbf{P}_{S_1}^{k}\mathbf{e}_{\mathcal{N}_1}} + \sum_{k=K}^{\infty}{\left(1-\frac{1}{d+1}\right)^k\mathbf{e}_{\mathcal{N}_1}^T\mathbf{P}_{S_1}^{k}\mathbf{e}_{\mathcal{N}_1}}\right]\\
&\leq 
\frac{1}{d+1}\left[K|\mathcal{N}_1|+\frac{\left(1-\frac{1}{d+1}\right)^{K+1}}{1-\left(1-\frac{1}{d+1}\right)}|\mathcal{N}_1|^2\frac{2}{n_1}\right]\\
&\leq \frac{x}{d+1}O(\log n_1) + \frac{2x^2}{n_1}.
\end{align*}

Finally, we set $n_1\geq 2d^3\geq 2d x^2$, the above inequality is bounded by $O(\frac{x\log n_1}{d})$. Combining this result with Eq. (\ref{equ:compute_lower_bound}), we obtain $r_{S_1}(s,t)=\frac{1}{x}+O(\frac{1}{dx}\log n_1 )$. This finishes the proof.\hfill$\blacksquare$\par

\stitle{Proof of Lemma \ref{lem:G1G2_differ}.}\hypertarget{proof_of_lemma:G1G2_differ}{}
    Recall that by our construction, for $\mathcal{G}_1$ all neighbors of $s$ is in $S_1$. Thus, $r_{\mathcal{G}_1}(s,t)=r_{S_1}(s,t)=\frac{1}{d_s}+O(\frac{\log n_1}{d d_s})=\frac{1}{d_s}+O(\frac{1}{d_s^2})$. This is equivalently saying $d_s r_{\mathcal{G}_1}(s,t) = 1+ O(\frac{1}{d_s})$. On the other hand, for $\mathcal{G}_2$, the neighbors of $s$ are $x=(1-\epsilon)d_s$ in $S_1$ and $\epsilon d_s$ in $S_2$. By the parallel resistance formula and Lemma \ref{lem:parallel_resistance},
    \begin{align*}
        r_{\mathcal{G}_2}(s,t) &= \frac{r_{S_1}(s,t)r_{S_2}(s,t)}{r_{S_1}(s,t)+r_{S_2}(s,t)} = \frac{\frac{2}{d_s-x}\left(\frac{1}{x}+O(\frac{1}{d_s x})\right)}{\frac{2}{d_s-x}+\frac{1}{x}+O(\frac{1}{d_s x})}\\
        &=\frac{2+O(\frac{1}{d_s})}{d+x+O(\frac{d_s-x}{d_s})}=\frac{2+O(\frac{1}{d_s})}{(2-\epsilon)d_s+O(\epsilon)}.
    \end{align*}

Thus, we have
\begin{align*}
   d_s r_{\mathcal{G}_2}(s,t) &= \frac{2d_s\left(1+O(\frac{1}{d_s})\right)}{(2-\epsilon)d_s+O(\epsilon)}= \frac{2}{2-\epsilon}\frac{1+O(\frac{1}{d_s})}{1+O(\frac{\epsilon}{d_s})}\\
   &=\frac{2}{2-\epsilon}\left(1+O(\frac{1}{d_s})\right).
\end{align*}

Finally, we set $d_s\gg \frac{1}{\epsilon}$, thus $d_s r_{\mathcal{G}_2}(s,t)=\frac{2}{2-\epsilon}\left(1+o(\epsilon)\right)$. Therefore, 
\begin{align*}
   | r_{\mathcal{G}_1}(s,t)-r_{\mathcal{G}_2}(s,t)|&=\frac{1}{d_s}\left[\frac{2}{2-\epsilon}-1+o(\epsilon)\right]=\frac{1}{d_s}\left[\frac{\epsilon}{2-\epsilon}+o(\epsilon)\right]\\
   &\geq \frac{\epsilon}{2}\frac{1}{d_s}\geq \frac{\epsilon}{4}r_{\mathcal{G}_1}(s,t).
\end{align*}

    The final inequality holds because $r_{\mathcal{G}_1}(s,t)=\frac{1}{d_s}+O(\frac{1}{d_s^2})\leq \frac{2}{d_s}$. This finishes the proof.\hfill$\blacksquare$\par

\stitle{Proof of Lemma \ref{lem:G1G2_distinguish}.}\hypertarget{proof_of_lemma:G1G2_distinguish}{}
Through our construction, we observe that the only way to distinguish between $\mathcal{G}_1$ and $\mathcal{G}_2$ is to perform a neighbor query on node $s$. For $\mathcal{G}_1$, we have all the neighbors of $s$, $\mathcal{N}_{\mathcal{G}_1}(s)\subset S_1$. For $\mathcal{G}_2$, we have $\mathcal{N}_1\subset S_1$ with $|\mathcal{N}_1|=(1-\epsilon)d_s$, $\mathcal{N}_2\subset S_2$ with $|\mathcal{N}_2|=\epsilon d_s$ and $\mathcal{N}_1\cup \mathcal{N}_2=\mathcal{N}_{\mathcal{G}_2}(s)$ is the neighbor set of $s$. Thus, distinguishing $\mathcal{G}_1$ and $\mathcal{G}_2$ via a neighbor query on $s$ reduces to determining whether the elements of $\mathcal{N}_{\mathcal{G}_1}(s), \mathcal{N}_{\mathcal{G}_2}(s)$ follow the same distribution. We prove this result using the classical Yao's Minimax Principle.

\begin{theorem}{(Yao's Minimax Principle, see e.g. Lemma D.1 from ~\cite{cai2023effective})}\label{thm:minimax_principle}
    Let $\mathcal{X}$ denotes the set of inputs of a problem, $\Delta(\mathcal{X})$ denotes the set of all distributions over $\mathcal{X}$. Let $\mathcal{A}_\mu$ denotes minimum cost among all deterministic algorithms that solves the problem with probability $\geq 2/3$, with respect to the distribution $\mu\in \Delta(\mathcal{X})$. Let $\mathcal{R}$ denotes the minimum cost among all randomized algorithms that solves the problem for all $x\in \mathcal{X}$. Then, $\mathcal{R}\geq \max_{\mu\in \mathcal{X}}{\mathcal{A}_\mu}$.
\end{theorem}

In our case, we choose the problem inputs $\mathcal{X}=\{\mathcal{N}(s)\}$, and we consider the following two events $C_1,C_2\subset \mathcal{X}$:
\begin{align*}
C_1=\{&\mathcal{N}(s):\forall v\in \mathcal{N}(s), v\in S_1\}\\
C_2=\{&\mathcal{N}(s):\mathcal{N}(s)=\mathcal{N}_1\cup \mathcal{N}_2, \\
&\mathcal{N}_1\subset S_1,
\mathcal{N}_2\subset S_2,\ and\  |\mathcal{N}_1|=(1-\epsilon)|\mathcal{N}(s)| \}
\end{align*}

Next, we consider the distribution $\mu$ over $\mathcal{X}$, such that: (i) $C_1$ happens with probability $\frac{1}{2}$ and $\forall x \in C_1\subset \mathcal{X}$ happens with equal probability; (ii) $C_2$ happens with probability $\frac{1}{2}$ and $\forall x \in C_2\subset \mathcal{X}$ happens with equal probability. We consider any property testing deterministic algorithm $A$, such that $A$ should answer YES if there are $\geq \epsilon|\mathcal{N}(s)|$ number of nodes $v\in \mathcal{N}(s)$ in $S_2$; $A$ should answer NO if there are $< \epsilon|\mathcal{N}(s)|$ number of nodes $v\in \mathcal{N}(s)$ in $S_2$. We denote $A(C_1)$, $A(C_2)$ as the random variable, which represents the answer of $A$ on some $x\in  C_1$ uniformly randomly (resp., $x\in  C_2$ uniformly randomly). We consider when $A$ succeeds with probability $\geq 2/3$:
\begin{align*}
    \frac{2}{3}\leq \mathbb{P}_{\mu}[A\ \ succeeds]=\frac{1}{2} \mathbb{P}[A(C_1)=NO]+\frac{1}{2} \mathbb{P}[A(C_2)=YES].
\end{align*}

Since $C_1$ is fixed, we only need to consider the second probability $\mathbb{P}[A(C_2)=YES]$. We assume $A$ makes $q$ queries on $C_2$. Then,
\begin{equation}\label{equ:AC2}
\begin{aligned}
    \mathbb{P}[A(C_2)=YES]&=\mathbb{P}[A(C_2)=YES | A(C_1)=A(C_2)]\mathbb{P}[A(C_1)=A(C_2)]\\
    &+\mathbb{P}[A(C_2)=YES | A(C_1)\neq A(C_2)]\mathbb{P}[A(C_1)\neq A(C_2)]\\
    &\leq \mathbb{P}[A(C_1)=YES]+\mathbb{P}[A(C_1)\neq A(C_2)]\\
    &= \mathbb{P}[A(C_1)=YES] + \mathbb{P}[one\ query\ returns\ v\in S_2 ]\\
    &\leq 1-\mathbb{P}[A(C_1)=NO] + \epsilon q.
\end{aligned}
\end{equation}

The final inequality holds because there are only $\epsilon|\mathcal{N}(s)|$ randomly chosen elements in $S_2$. For each fixed query, the probability that we find some $v\in S_2$ is at most $\epsilon$. Arming with Eq. (\ref{equ:AC2}), we have:
\begin{align*}
    \frac{2}{3}&\leq \frac{1}{2} \mathbb{P}[A(C_1)=NO]+\frac{1}{2} \mathbb{P}[A(C_2)=YES]\leq \frac{1}{2}+\frac{1}{2}\epsilon q.
\end{align*}

As a result, $q\geq \frac{1}{3\epsilon}$. This implies $A$ makes at least $\Omega(\frac{1}{\epsilon})$ queries on $C_2$. By Yao's Minimax Principle, any randomized algorithm with success probability $\geq 2/3$ that distinguish $\mathcal{N}_{\mathcal{G}_1}(s)$ and $\mathcal{N}_{\mathcal{G}_2}(s)$ requires $\Omega(\frac{1}{\epsilon})$ queries. This finishes the proof of Theorem \ref{thm:lower_bound_specific}.\hfill$\blacksquare$\par

\end{document}